%% file: FS1.tex
\documentclass[12pt]{article} 
\usepackage{amssymb}
\usepackage{times}
\usepackage{epsfig}
\usepackage{xr}
\newtheorem{theorem}{Theorem}[section]

\newtheorem{lemma}[theorem]{Lemma}
\newtheorem{proposition}[theorem]{Proposition}
\newtheorem{corollary}[theorem]{Corollary}
\newtheorem{remark}[theorem]{Remark}
\newcommand{\proofendsign}{\hfill\hbox{\vrule width 7pt depth 0pt height 7pt} \par\vspace{10pt}}
\newenvironment{proof}{\par\noindent {\it Proof:} \hspace{7pt}}%
{\proofendsign}
\newenvironment{mainproof}{\par\noindent {\it Proof of Theorem \ref{thm:main}:} \hspace{7pt}}%
{\proofendsign}

\newcommand{\rv}{{\rm v}}

\input lemsmac

\input figMac
\def\figdir{fig/}

\newcommand{\set}[2]{\big\{ \ #1\ \big|\ #2\ \big\}}

\newcommand{\E}{{\rm e}}
\newcommand{\I}{{\rm i}}
\newcommand{\dd}{{\rm d}}

\newcommand{\Hap}{{\bf H1'}}
\newcommand{\Hbp}{{\bf H2'}}
\newcommand{\Hcp}{{\bf H3'}}
\newcommand{\Hdp}{{\bf H4'}}

\renewcommand{\const}{\mbox{ \rm const }}
\newcommand{\bx}{{\bf x}}
\newcommand{\bk}{{\bf k}}
\newcommand{\bzer}{{\bf 0}}
\newcommand{\ba}{{\bf a}}
\newcommand{\bb}{{\bf b}}
\newcommand{\bc}{{\bf c}}
\newcommand{\bt}{{\bf t}}
\newcommand{\bp}{{\bf p}}
\newcommand{\bq}{{\bf q}}
\newcommand{\br}{{\bf r}}

\newcommand{\smint}{{\tst \int}}
\newcommand{\cst}[2]{{\rm const}^{#1}_{#2}}

\newcommand{\bnabla}{\nabla}
\renewcommand{\vol}{\mbox{ \rm  Vol}}
\newcommand{\implies}{\Longrightarrow}
\newcommand{\Item}{\medskip\noindent}
\newcommand{\rP}{{\rm P}}

\newcommand{\pmb}[1]{\setbox0=\hbox{#1}       
     \kern-.025em\copy0\kern-\wd0
     \kern.05em\copy0\kern-\wd0
     \kern-.025em\box0}             

\newcommand{\bphi}{{\mathchoice{\pmb{$\phi$}}{\pmb{$\phi$}}
                              {\pmb{$\scriptstyle\phi$}}
                              {\pmb{$\scriptscriptstyle\phi$}}}}
\newcommand{\bth}{{\mathchoice{\pmb{$\th$}}{\pmb{$\th$}}
                              {\pmb{$\scriptstyle\th$}}
                              {\pmb{$\scriptscriptstyle\th$}}}}

\newcommand{\hypflatnest}{NN}
\newcommand{\dblInt}{\mathop{\int\!\!\!\int}}

\begin{document}

\title{Singular Fermi Surfaces I. General Power Counting and Higher Dimensional
Cases}

\author{Joel Feldman$^1$\footnote{feldman@math.ubc.ca; supported by NSERC of Canada} 
 and Manfred Salmhofer$^{1,2}$\footnote{salmhofer@itp.uni-leipzig.de; supported by DFG--grant Sa-1362/1-1, an ESI senior research fellowship, and NSERC of Canada} \\
\small  $^1$ Mathematics Department, The University of British Columbia \\
\small 1984 Mathematics Road, Vancouver, B.C., Canada V6T 1Z2
\\
\small $^2$ Theoretische Physik, Universit\"{a}t Leipzig \\
\small Postfach 100920, 04009 Leipzig, Germany
}

\date{\today}

\maketitle

\begin{abstract}
\noindent
We prove regularity properties of the self--energy, 
to all orders in perturbation theory, for systems with 
singular Fermi surfaces which contain
{\em Van Hove points} where the gradient of the dispersion relation vanishes. 
In this paper, we show for spatial dimensions $d \ge 3$ 
that despite the Van Hove singularity, the overlapping loop bounds
we proved together with E.\ Trubowitz for regular non--nested Fermi surfaces
[J.\ Stat.\ Phys.\ 84 (1996) 1209]
still hold, provided that the Fermi surface satisfies a no--nesting condition.
This implies that for a fixed interacting Fermi surface, 
the self-energy is a continuously differentiable function of 
frequency and momentum, so that the quasiparticle weight and 
the Fermi velocity remain close to their values in the noninteracting 
system to all orders in perturbation theory.
In a companion paper, we treat the more singular two--dimensional case.
\end{abstract}

\section{Introduction}
In 1953, Van Hove published a general argument implying the occurrence 
of singularities in the phonon and electron spectrum of  crystals \cite{VanHove}.
The core of  his argument is an application of Morse theory \cite{Morse} ---
a sufficiently smooth function defined on the torus and having only 
nondegenerate critical points must have saddle points. 

In the independent--electron approximation, 
the dispersion relation $k \mapsto \epsilon (k)$ of the electrons plays 
the role of the Morse function, 
and the {\em Van Hove singularities} (VHS) manifest themselves in 
the electronic density of states 
\begin{equation}
\rho (E) = \int \frac{\dd^{d} k}{(2\pi)^d}  \; \de (E - \epsilon (k))
\end{equation}
at those values of the energy 
where the level set $\{ k : \epsilon (k) = E\}$ 
contains one (or more) of the saddle points, 
the  {\em Van Hove points}. 
The nature of these singularities in $\rho$  depends on the dimension. 
In two dimensions, $\rho$ has a logarithmic singularity. In three dimensions, 
$\rho$ is continuous but its derivative has singularities. 
In all dimensions, these singularities have observable consequences,
although they occur only at discrete values of the energy. 

In mean--field theories for symmetry breaking, the density of states plays an important role
because it enters the self--consistency equations for the order parameter. 
For instance, in BCS theory, the superconducting gap $\Delta$
is determined as a function of the temperature $T = \beta^{-1}$ 
as the solution to the equation 
\begin{equation}\label{eq:BCS}
\Delta
=
g \;
\Delta 
\int \dd E \; \frac{\rho(E)}{2\sqrt{(E-E_F)^2+\Delta^2}} \; 
\tanh \frac{\beta \sqrt{(E-E_F)^2+\Delta^2}}{2}
\end{equation}
where $g>0$ is the coupling constant that determines the strength of the 
mean--field interaction between Cooper pairs and $E_F$ is the Fermi energy 
determined by the electron density. (We have written the equation for 
an $s$--wave superconductor.)
The properties of $\rho(E)$ for $E$ near to $E_F$ obviously 
influence the temperature--dependence of $\Delta$, 
as well as the value of the critical 
temperature $T_c$, defined as the largest value of $T$ below which \Ref{eq:BCS}
has a nonzero solution.
If $\rho$ is smooth, the small--$g$ asymptotics of $T_c$ is 
$T_c \sim \E^{- \rho(E_F)/g}$. A logarithmic divergence in $\rho$ of the form
$\rho(E) = K \ln \frac{W}{|E-E_{VH}|}$ (with fixed constants $K$ and $W$) enhances the
critical temperature to $T_c \sim \E^{-K/\sqrt{g}}$ if $E_F = E_{VH}$.
Similarly, van Hove singularities cause ferromagnetism in mean--field theory
at arbitrarily small couplings $g \ll 1$. 

In a true many--body theory, all this becomes much less clear--cut. 
Besides the obvious remark that in two dimensions, there is no 
long--range order at positive temperatures \cite{Koma}, 
hence the above discussion is restricted to mean--field theory, 
the question whether Van Hove singularities
indeed occur in interacting systems and if so, what their influence on 
observable quantities is, remains open and important. 
The theoretical quantity related to the electron spectrum and the density of states 
of the interacting
system is the interacting dispersion relation or the spectral function, 
obtained from the full propagator, hence ultimately from the electron self--energy.
The VHS might cease to exist in the interacting system for various reasons. 
The interacting Fermi surface may turn out 
to avoid the saddle points, or the singularity caused by the saddle points 
of the dispersion relation may be smoothed out by more drastic effects, 
such as the opening of gaps in the vicinity of the saddle points. 
On the other hand, the VHS might also become more generic because the Fermi surface 
may get pinned at the Van Hove points, and the singularity might also get
stronger due to interaction effects. 
A lot of research has gone into these questions because Van Hove singularities
were invoked as a possible explanation of high--temperature superconductivity
(see, e.g.\ \cite{Markiewicz} and references therein). In particular, there are 
competition effects between superconductivity, 
ferromagnetism and antiferromagnetism \cite{HM,HSFR,HS},
as well as interesting phenomena connected to Fermi surface fluctuations
\cite{MeNeu,MeQC}, to mention but a few results. The above speculations 
as to the fate of the Fermi surface and the VHS
have been discussed widely in the 
literature \cite{Markiewicz}. 

In this paper, we begin a mathematical study 
of Fermi surfaces that contain Van Hove points, but that satisfy a no--nesting 
condition away from these points, with the aim of understanding  
some of the above questions.  
We prove regularity properties of the electron self--energy 
to all orders of perturbation theory using the multiscale techniques
of \cite{FST1,FST2,FST3,FST4}, which are 
closely related to the renormalization group techniques used in \cite{HM,HSFR}.
In the present paper, we give bounds that apply in all dimensions $d \ge 2$
and then consider the case $d \ge 3$ in more detail. In a companion paper 
\cite{SFS2}, we focus
on the two--dimensional case, and in particular on the question of the 
renormalization of the quasiparticle weight and the Fermi velocity.

Our motivation for imposing the no--nesting condition is twofold.
First, an example of a dispersion relation in $d=2$ with a Fermi surface that contains 
Van Hove points and satisfies our no--nesting condition is the $(t,t')$ Hubbard model
with $t' \ne 0$ and $t \; t' < 0$ at the Van Hove density.
For  $t' = 0$, the Van Hove density is at half--filling, and the 
Fermi surface becomes flat, hence nested under our definition. 
However, there is ample evidence that in the Hubbard  model 
it is the parameter range $t' \ne 0$ and  electron density near to the van Hove density  
that is relevant  for high--$T_c$ superconductivity 
(see, e.g.\ \cite{Markiewicz,HM,HSFR}). 
Second, nesting causes additional singularities, 
and to get a clear picture of which property of the Fermi surface
causes what kind of phenomena, 
it is useful to disentangle the effects of the VHS from those of nesting. 

We now give an overview of the technical parts of the present paper
and state our main result about the self--energy and the correlation functions.
In Section \ref{genpoco}, we prove bounds for volumes of thin shells in momentum 
space close to the Fermi surface. These volume bounds are the essential 
ingredient for power counting bounds. In Lemma \ref{le:2.4}, we show that 
these volume bounds are not changed by the introduction of the most common
singularities in $d \ge 3$ and increase by a 
logarithm of the scale in $d=2$. This implies by the general bounds of 
\cite{FST1} that the superficial power counting of the model is unchanged 
for $d \ge 3$ and changes ``only'' by logarithms in $d=2$.
Lemma \ref{le:2.6} contains a refinement of these bounds 
in which one restricts to small balls near the singular points. 
In Section \ref{finer}, we turn to the finer aspects of power counting that 
are necessary to understand the regularity of the self--energy, 
for spatial dimensions $d \ge 3$.  
We define a weak no--nesting condition which is essentially identical to that
of \cite{FST1} and prove that the volume improvement estimate 
(1.34) of \cite{FST1} carries over unchanged (Proposition \ref{propOverlap}).
By Theorem 2.40 of \cite{FST1}, this implies that the bulk of the 
conclusions of Theorems 1.2 -- 1.8 of \cite{FST1} carry over to the situation 
with VHS in $d \ge 3$. Namely,

\begin{theorem}\label{thm:main}
Let $d \ge 3$, and let the dispersion relation $\bk \mapsto e(\bk)$ satisfy

\begin{itemize}

\item[$\circ$] $\cF=\set{\bk\in\bR^d}{e(\bk)=0}$ is compact

\item[$\circ$] $e(\bk)$ is $C^3$ 

\item[$\circ$] $\nabla e(\bk)$ vanishes only at isolated points of $\cF$.
We shall call them singular points.

\item[$\circ$] if $e(\bk)=0$ and $\bnabla e(\bk)=\bzer$, then
           $\big[\sfrac{\del^2\hfill}{\del\bk_i\del\bk_j} e(\bk)
                \big]_{1\le i,j\le d}$ is nonsingular and has at least 
                one positive eigenvalue and at least one negative eigenvalue.

\item[$\circ$]There is no nesting, in the precise sense of
Hypothesis \hypflatnest\ in Section \ref{sec:NN} 
\end{itemize}

\noindent
Let the interaction be short--range in the sense that the Fourier transform 
$k \mapsto \hat v(k)$ of the two--body interaction is twice continuously 
differentiable in  $k$. Introduce the counterterm function $\bk \mapsto K(\bk)$ 
as in Section 2 of \cite{FST1}, but using the localization operator 
$(\ell T)(q_0,\bq)=T(0,\bq)$ in place of the localization operator of 
Definition 2.6 of \cite{FST1}, to renormalize the perturbation expansion.  Then

\begin{enumerate}

\item
To any fixed order in renormalized perturbation theory, the electronic 
self--energy (i.e. the sum of the values of all two--legged 
one--particle--irreducible Feynman graphs) is continuously differentiable 
in the frequency and momentum variables.
There is an $\veps > 0$ so that all first derivatives
of the self--energy are H\" older continuous of degree $\veps$.
The counterterm function $K$ has the same regularity properties. 

\item
To any fixed order in renormalized perturbation theory, 
all correlation functions are well--defined, locally integrable 
functions of the external momenta. 

\item
To any fixed order in renormalized perturbation theory, 
the only contributions to the four--point function that fail to be bounded
and continuous are the (generalized) ladder diagrams 
$$
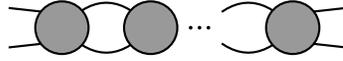
\figput{ladder}
$$
Here each vertex \figplace{vertex}{0 in}{-0.07 in} is an arbitrary connected 
four--legged subdiagram and each line \figplace{line}{0 in}{-0.07 in} 
is a string \figplace{string}{0 in}{-0.08 in} whose vertices (if any) are arbitrary one particle irreducible
two--legged subdiagrams.

\item
For each natural number $r$, denote by $\la^rK_r(\bk)$ the order $r$
contribution to the renormalized perturbation expansion of the 
counterterm function $K(\bk)$.
For each natural number $R$, the map $e \mapsto e+\sum_{r=1}^R \la^r K_r$ 
is locally injective. The precise meaning of and hypotheses for this 
statement are given in the paragraph containing \Ref{eq:injective}.
\end{enumerate}
\end{theorem}

\noindent
The above statements are proven in Section \ref{sec:mainproof} at 
temperature $T=0$. However, the same methods show that they 
extend to small $T \ge 0$, with the change that for $T>0$,
singularities are replaced by finite values that, however, diverge as 
$T \to 0$. 

As explained in detail in \cite{FST1}, the counterterm $K$ fixes the 
Fermi surface, so that all our results are about the model with 
a fixed interacting Fermi surface. Whether the situation that 
the Fermi surface contains zeroes of the gradient of $e$ can 
indeed be achieved is related to the question whether there 
is an inversion theorem generalizing that of \cite{FST4} to 
the situation with VHS, i.e.\ which provides existence of a solution 
of the equation $e + K(e) = E$ for the present situation (item 4 
of the above theorem only gives local uniqueness).
This is a difficult question which is 
still under investigation (see also \cite{SFS2}).

A natural question is the relation between these statements 
to all orders in perturbation theory and results obtained from 
truncated renormalization group flows in applied studies, 
which are often claimed to be ``nonperturbative''. 
The all--order results are statements about an iterative solution 
to a full renormalization group flow. The solution of renormalization 
group flows obtained by truncating the infinite hierarchy to a finite hierarchy
creates scale--dependent 
approximations to the Green functions. These approximations give
the leading order behaviour if the truncation has been done appropriately. 
Often, the results indicate 
instabilities of the flow, which signal that the true state of the system 
is not well--described by an action of the form assumed in the flow. 
A true divergence in the solution occurs only when the regime of validity
of the truncation is left. (In the simplest situations, such singularities
coincide with the divergence of a geometric series.) 
In careful studies, the equations are never integrated to the point where
anything diverges. In that case, the regularity bounds obtained by all--order 
estimates are more accurate than those obtained from 
the solution of the flow equations truncated at finite order. 
There is one case where the integration of the 
renormalization group equations gives an effect within the validity of the 
truncation, but qualitatively different from all--order theory: this is 
when the flow satisfies infrared asymptotic freedom, i.e.\ the 
coupling function becomes screened at low scales. 
For instance, in the repulsive Hubbard model, the ladders with the
bare interaction lead to a screening of the superconducting interaction,
corresponding to $g<0$ and hence to no solution in the BCS gap equation.
(However, in this case, an attractive Cooper interaction is generated in 
second order, and it then grows in the flow to lower scales.) 
Such screening effects can only make terms smaller. Hence the 
upper bounds provided by the all--order analysis are still as good as the 
integration of truncations to the same order, as far as regularity properties
are concerned. In practice, the truncations done in the RG equations 
are of very low order, so that the all--order analysis includes many 
contributions that are not taken into account in these truncations. 

A {\em nonperturbative} mathematical proof 
involves bounding the remainders created in the expansion 
(or truncation). This is possible in $d=2$ using the sector method of 
\cite{FMRT}, but a full construction has not yet been achieved in 
$d \ge 3$. Because the graphical structures used in our arguments 
only require one overlap of loops, we expect that a suitable adaptation 
will be possible in constructive studies. In addition to 
the above--mentioned problem with constructive arguments in $d \ge 3$, 
the important question of the inversion theorem should also be addressed.

\section{General power counting bounds in $d \ge 2$}\label{genpoco}

\subsection{Analytic structure of the one--body problem}
Here we discuss briefly the properties of the one--body problem, 
to show that the Fermi surface of the noninteracting system is 
given as the zero set of an analytic function, hence no-nesting 
in a polynomial sense is a generic condition. For lattice models,
analyticity of the dispersion relation $e$ is obvious for hopping 
amplitudes that decay exponentially with distance (or are even
of finite range). For continuum Schr\" odinger operators, 
it follows from the statements below, which even hold for the 
case with a magnetic field. 
 
Let $d\ge 2$ and $\Gamma$ be a lattice in $\bR^d$ of maximal rank.  Let $r>d$. Define
\begin{eqnarray}
\cA  
&=& 
\set{ A=(A_1,\ldots,A_d) \in \big( L_\bR^r(\bR^d/\Gamma) \big)^d } {\smint_{\bR^d/\Gamma} A(x) dx =0} 
\nonumber\\
\cV  &=& 
\set{V  \in L_\bR^{r/2}(\bR^d/\Gamma)} {\smint_{\bR^d/\Gamma} V(x) dx =0} 
\nonumber
\end{eqnarray}
For $(A,V)\in \cA\times \cV $ set
$$
H_\bk(A,V) = \big( i\bnabla + A(x) - \bk \big)^2 + V(x)
$$ 
When $d=2,3$, the operator $H_\bk(A,V)$ describes an electron in $\bR^d$
with quasimomentum $\bk$ moving under the influence of the magnetic field with periodic vector potential $A(x)=(A_1(x),\ldots A_d(x))$ and electric field 
with periodic potential $V(x)$.
The conditions $\smint_{\bR^d/\Gamma} A(x) dx =0 $ and 
$\smint_{\bR^d/\Gamma} V(x) dx =0$ are included purely for convenience
and can always be achieved by translating $\bk$ and shifting the zero point
of the energy scale. 
The following theorem is proven in  \cite{FKTasy}.

\begin{theorem}
Let
\begin{eqnarray}
\cA_\bC &=& 
\set{ A=(A_1,\ldots,A_d) \in \big( L^r(\bR^d/\Gamma) \big)^d } {\smint_{\bR^d/\Gamma} A(x) dx =0} \nonumber\\
\cV_\bC &=& \set{V  \in L^{r/2}(\bR^d/\Gamma)} {\smint_{\bR^d/\Gamma} V(x) dx =0}  
\nonumber
\end{eqnarray}
be the complexifications of $\cA$ and $\cV$ respectively.
There exists an analytic function $F$ on 
$\bC^d \times \bC \times \cA_\bC \times \cV_\bC$
such that, for $\bk,A,V$ real,
$$
\lambda \in {\rm Spec}\big( H_\bk(A,V)\big) \qquad \iff \qquad F(\bk,\lambda,A,V)= 0
$$
\end{theorem}

\noindent
The theorem is proven by providing a formula for $F$. Write
$$
\big(i\bnabla+A(x)-\bk\big)^2+V(x)-\la=\bbbone-\De+u(\bk,\la)+w(\bk,A,V)
$$
with
\begin{eqnarray}
u(\bk,\la)
&=&
-2i\bk\cdot\bnabla+\bk^2-\la-\bbbone
\nonumber\\
w(\bk,A,V)
&=&
i\bnabla\cdot A+iA\cdot\bnabla-2\bk\cdot A+A^2+V
\nonumber
\end{eqnarray}
Then the function
$
F(\bk,\lambda,A,V)
$
of the above theorem is a suitably regularized determinant of 
$\bbbone + \sfrac{1}{\sqrt{\bbbone-\De}}
u(\bk,\la)\sfrac{1}{\sqrt{\bbbone-\De}}
+ \sfrac{1}{\sqrt{\bbbone-\De}}w(\bk,A,V)\sfrac{1}{\sqrt{\bbbone-\De}}$.

\subsection{Volumes of Shells around Singular Fermi Surfaces}\label{sec:2.2}
Suppose that the energy eigenvalues for the one--body problem with quasimomentum
$\bk$ are the solutions of an equation $F(\bk,\la)=0$. That is, the bands
$e_1(\bk)\le e_2(\bk)\le e_3(\bk)\le\cdots$ all obey
$F\big(\bk,e_n(\bk)\big)=0$. Our analysis of the regularity properties
of the self--energy and correlation functions depends on having good bounds
on the volume of the set of all quasimomenta $\bk$ for which there are
very low energy bands. More precisely, fix any $M>1$ and let $j\le 0$.
We need to know the volume of the set of all quasimomenta $\bk$ for which
there is at least one band with $\big|e_n(\bk)\big|\le M^j$. The following
lemma provides a useful simplification.

\begin{lemma}
Let $\cK$ be a compact subset of $\bR^d$ and
$F:\cK\times[-1,1]\rightarrow\bR$ be $C^1$. 
Then there is a constant $C$ such that
\begin{eqnarray}
&&\vol\set{\bk\in \cK}{F(\bk,\la)=0\hbox{ for some } |\la| \le M^j}\nonumber\\
&&\hskip2in\le \vol\set{\bk\in \cK}{|F(\bk,0)|\le C M^j}
\nonumber
\end{eqnarray}
for all $j\le 0$. In particular, if all bands $e_n(\bk)$ obey
$F\big(\bk, e_n(\bk)\big)=0 $ then,
$$
\vol\set{\bk\in \cK}{ |e_n(\bk)| \le M^j\hbox{ for some }n}
\le \vol\set{\bk\in \cK}{|F(\bk,0)|\le C M^j}
$$
\end{lemma}

\begin{proof} Since  $F$ is $C^1$ on the compact set $\cK\times[-1,1]$,
$$
C\equiv\sup_{(\bk,\la)\in \cK\times[-1,1]}
\big|\sfrac{\partial F}{\partial\la}(\bk,\la)\big|<\infty
$$
Hence, if for some $\bk\in \cK$ and some $|\la|\le M^j$, we have 
$F(\bk,\la)=0$, then, for that same $\bk$,
$$
|F(\bk,0)|=\big|F(\bk,\la)-F(\bk,0)\big|\le C|\la|\le CM^j
$$
Hence
$$
\set{\bk\in \cK}{F(\bk,\la)=0\hbox{ for some } |\la| \le M^j}
\subset\set{\bk\in \cK}{|F(\bk,0)|\le C M^j}
$$
\end{proof}\goodbreak

We now, and for the rest of this paper, focus on a single band $k \mapsto \epsilon (k)$,
and assume that the chemical potential $\mu$, used to fix the density, 
is such that $e(k) = \epsilon (k) -\mu$ has a nonempty zero set, 
the Fermi surface, which has also not degenerated to a point. 

In the scale analysis, momentum space is cut up in shells around the 
Fermi surface. Here we take the convention
of labelling these shells by negative integers $j \le 0$. The shell number $j$
contains momenta $k$ and Matsubara frequencies $k_0$ with 
$\frac12 M^j \le |\I k_0 - e(k)| \le M^j$. Here $M>1$ is fixed once and for all. 
For the (standard) details about the scale decomposition and the 
corresponding renormalization group flow, obtained by integrating 
over degrees of freedom in the shell number $j$ successively, 
downwards from $j=0$,  
see, e.g.\ \cite{FST1}, Section 2.

The next lemma contains the basic volume bound for the scale analysis. 
In the case without VHS, the bound is of order $M^j$. The lemma implies
that this bound remains unchanged for $d \ge 3$, and that there is an 
extra logarithm in $d=2$.
 
\begin{lemma}\label{le:2.4}
Let $\cK$ be a compact subset of $\bR^d$ and
$e:\cK\rightarrow\bR$ be $C^2$. Assume that for every point $\bp\in \cK$
at least one of
\item{$\circ$} $e(\bp)\ne 0$
\item{$\circ$} $\bnabla e(\bp)\ne 0$
\item{$\circ$} $\det\big[
                   \sfrac{\partial^2\hfill}{\partial\bk_i\partial\bk_j} e(\bp)
                \big]_{1\le i,j\le d}\ne 0$

\noindent is true. Then there is a constant $C$ such that
$$
\vol\set{\bk\in \cK}{|e(\bk)|\le M^j} \le C M^j\cases{|j|& if $d=2$\cr
                                                     1 & if $d>2$
}
$$
for all $j\le -1$.
\end{lemma}

\begin{proof} Since $\cK$ is compact, it suffices to prove that, for each $\bp\in \cK$
there are constants $R>0$ and $C$ (depending on $\bp$) such that for all $j\le -1$,
\begin{eqnarray}
\cV_{R,j}(\bp)
&=&
\vol\set{\bk\in \cK}{|e(\bk)|\le M^j, |\bk-\bp|\le R} 
\nonumber\\
&\le &
C M^j\cases{|j|& if $d=2$\cr
                                                     1 & if $d>2$
}
\nonumber
\end{eqnarray}

\noindent{\it Case 1: }$e(\bp)\ne 0$. We are free to choose $R$ sufficiently small
that $|e(\bk)|\ge\sfrac{1}{2} |e(\bp)|$ for all $\bk\in \cK$ with 
$|\bk-\bp|\le R$. Then $\set{\bk\in \cK}{|e(\bk)|\le M^j, |\bk-\bp|\le R}$
is empty unless $M^j\ge \sfrac{1}{2} |e(\bp)|$ and it suffices to take
$$
C=\sfrac{2}{|e(\bp)|}\vol\set{\bk\in \cK}{|\bk-\bp|\le R}
$$
\noindent{\it Case 2: }$e(\bp)= 0$, $\bnabla e(\bp)\ne 0$. By translating
and permuting indices, we may assume that $\bp=0$ and that 
$\sfrac{\partial\, e\,}{\partial \bk_1}(\bp)\ne 0$. Then, if $R$ is small enough,
$$
\bx= X(\bk)=\big(e(\bk),\bk_2,\ldots,\bk_d\big)
$$
is a $C^2$ diffeomorphism from $\cK_R=\set{\bk\in \cK}{|\bk-\bp|\le R}$ to 
some bounded subset $\cX$ of $\bR^d$. The Jacobian of this diffeomorphism is
$\sfrac{\partial\,e\,}{\partial \bk_1}(\bk)$ and is bounded away from zero,
say by $c_1$. Then
\begin{eqnarray}
\cV_{R,j}(\bp)
&=&
\vol\set{\bk\in \cK}{|e(\bk)|\le M^j, |\bk-\bp|\le R}
\nonumber\\
&=&
\int_{\cK_R}1\big( |e(\bk)|\le M^j\big)\ d^d\bk
\nonumber\\
&=&
\int_{\cX}1\big( |\bx_1|\le M^j\big)\ 
  \big|\sfrac{\partial\,e\,}{\partial \bk_1}\big(X^{-1}(\bx)\big)\big|^{-1}\ d^d\bx
\nonumber\\
&\le&
 c_1^{-1}\int_{\cX}1\big( |\bx_1|\le M^j\big)\ d^d\bx
\nonumber\\
&\le&
 c_1^{-1}c_2 M^j 
\nonumber
\end{eqnarray}
Here $1(E)$ denotes the indicator function of the event $E$, i.e.
$1(E) = 1$ if $E$ is true and $1(E)=0$ otherwise.

\noindent{\it Case 3: }$e(\bp)= 0$, $\bnabla e(\bp)= 0$, 
$\det\big[
        \sfrac{\partial^2\hfill}{\partial\bk_i\partial\bk_j} e(\bp)
           \big]_{1\le i,j\le d}\ne 0$. By translating, we may 
assume that $\bp=0$. Then, if $R$ is small enough, the Morse lemma \cite[Theorem
8.3bis]{Morsebis}
implies
that there exists a $C^1$ diffeomorphism, $X(\bk)$, from $\cK_R$ to some bounded
subset $\cX$ of $\bR^d$ such that 
$$
e\big(X^{-1}(\bx)\big)
=Q_m (\bx) 
=x_1^2+\ldots+x_m^2-x_{m+1}^2-\ldots-x_d^2
$$
for some $0\le m\le d$. Then
\begin{eqnarray}
\cV_{R,j}(\bp)
&=&
\int_{\cK_R}1\big( |e(\bk)|\le M^j\big)\ d^d\bk
\nonumber\\
&=&\int_{\cX}
    1\big( |Q_m (\bx) |\le M^j\big)\ 
  \Big|\det\big[\sfrac{\partial X_i}{\partial\bk_j} \big(X^{-1}(\bx)\big)
                \big]_{1\le i,j\le d}\Big|^{-1}d^d\bx
\nonumber\\
&\le& c_1^{-1}\int_{\cX}
    1\big( |Q_m (\bx) |\le M^j\big)\ d^d\bx
\nonumber
\end{eqnarray}
If $m=0$ or $m=d$,
\begin{eqnarray}
\int_{\cX}
    1\big( |Q_m (\bx) |\le M^j\big)\ d^d\bx
&=&
\int_{\cX}
    1\big( |x_1^2+\ldots+x_d^2|\le M^j\big)\ d^d\bx
\nonumber\\
&\le&
 \int_{\bR^d}
    1\big( |x_1^2+\ldots+x_d^2|\le M^j\big)\ d^d\bx
\nonumber\\
&=&
c_d M^{dj/2}
\nonumber
\end{eqnarray}
so it suffices to consider $1\le m\le d-1$. Go to spherical coordinates
separately in $x_1,\ldots,x_m$ and $x_{m+1},\ldots,x_d$, using
$$
u=\sqrt{x_1^2+\ldots+x_m^2}\qquad
v=\sqrt{x_{m+1}^2+\ldots+x_d^2}
$$
If $R$ is small enough
$$
\cV_{R,j}(\bp)
\le c_1^{-1}c_{m,d}\int_{0\le u,v\le 1}\hskip-26pt
    1\big( |u^2-v^2|\le M^j\big)u^{m-1}v^{d-m-1}\ dudv
$$
Now make the change of variables $x=u+v$, $y=u-v$. Then
\begin{eqnarray}
\int_{0\le v\le u\le 1}\hskip-35pt &&
    1\big( |u^2-v^2|\le M^j\big)u^{m-1}v^{d-m-1}\ dudv
\nonumber\\
&\le&
 \int_0^2\hskip-8pt dx\int_0^1\hskip-8pt dy\ 
    1\big( xy\le M^j\big)(x+y)^{m-1}|x-y|^{d-m-1}
\nonumber\\
&\le&
 \int_0^2\hskip-5pt dx\int_0^1\hskip-5pt dy\ 
    1\big( xy\le M^j\big)(x+y)^{d-2}
\nonumber
\end{eqnarray}
and the lemma follows from 
$$
\int_0^2\!\!dx\int_0^{\min\{1,M^j/x\}}\!\! dy
=
M^j+M^j\ln\sfrac{2}{M^j}
$$
and, for $n \ge 1$, 
$$
\int_0^2\!\!dx\int_0^{\min\{1,M^j/x\}}\!\! dy\ x^n
\le
 \sfrac{1}{n+1}M^{(n+1)j}+2^{n} M^j
$$
and
$$
\int_0^2\!\!dx\int_0^{\min\{1,M^j/x\}}\!\! dy\ y^n
\le M^j
$$
\end{proof}

\noindent
That $\vol\set{\bk\in \cK}{|e(\bk)|\le M^j} \le C M^j|j|$ and that this
bound suffices to yield a well--defined counterterm and well--defined
correlation functions, to all orders of perturbation theory, was also
proven in \cite{Brox}. We now refine Lemma \ref{le:2.4} a little.

\goodbreak

\begin{lemma}\label{le:2.6}
Let $e:\bR^d\rightarrow\bR$ be $C^2$. Assume that
\begin{itemize}
\item[$\circ$] $e(\bzer)= 0$
\item[$\circ$] $\bnabla e(\bzer) = \bzer$
\item[$\circ$] $\det\big[
                   \sfrac{\partial^2\hfill}{\partial\bk_i\partial\bk_j} e(\bzer)
                \big]_{1\le i,j\le d}\ne 0$
\item[$\circ$] $\big[
                   \sfrac{\partial^2\hfill}{\partial\bk_i\partial\bk_j} e(\bzer)
                \big]_{1\le i,j\le d}$ has at least one positive eigenvalue
and at least one negative eigenvalue.
\end{itemize}
Then there are $C,\ C'>0$ such that for all $\bq \in\bR^d$,
$j\le 0$ and $0<\veps<\sfrac{1}{2}$,
\begin{eqnarray}
&&
\vol\set{\bk\in \bR^d}{|e(\bk)|\le M^j, |\bk-\bq|\le M^{\veps j}, |\bk|\le C'}
\nonumber\\
&\le &
CM^j\cases{1+(1-2\veps)|j|&if $d=2$\cr
               M^{(d-2)\veps j}&if $d>2$\cr
}
\nonumber
\end{eqnarray}
\end{lemma}

\begin{proof}
By the Morse lemma, we can assume without loss of generality that
$$
e\big(\bk\big)=k_1^2+\ldots+k_m^2-k_{m+1}^2-\ldots-k_d^2
$$
for some $1\le m\le d-1$. Go to spherical coordinates
separately in $k_1,\ldots,k_m$ and $k_{m+1},\ldots,k_d$, using
$$
u=\sqrt{k_1^2+\ldots+k_m^2}\qquad
v=\sqrt{k_{m+1}^2+\ldots+k_d^2}
$$
For any fixed $u>0$, the condition $|\bk-\bq|\le M^{\veps j}$ restricts
$(k_1,\ldots,k_m)$ to lie on a spherical cap of diameter at most $2M^{\veps j}$
on the sphere of radius $u$. This cap has an area of at most an $m$--dependent
constant times $\min\{u, M^{\veps j}\}^{m-1}$.  Similarly,
for any fixed $v>0$, the condition $|\bk-\bq|\le M^{\veps j}$ restricts
$(k_{m+1},\ldots,k_d)$ to run over an area of at most a
constant times $\min\{v, M^{\veps j}\}^{d-m-1}$.  The condition 
$|\bk-\bq|\le M^{\veps j}$ also restricts $u$ and $v$ to run over intervals
$I_1$, $I_2$ of length at most $2M^{\veps j}$. Thus
\begin{eqnarray}
&&\vol\set{\bk}{|e(\bk)|\le M^j, |\bk-\bq|\le M^{\veps j}}\cr
&&\le
 c_{m,d}\int_{I_1\times I_2}
    1\big( |u^2-v^2|\le M^j\big)
        \min\{u, M^{\veps j}\}^{m-1}\min\{v, M^{\veps j}\}^{d-m-1}\ dudv
\nonumber\\
&&\le
c_{m,d}\int_{I_1\times I_2}
    1\big( |u^2-v^2|\le M^j\big)
        \min\Big\{\max\{u,v\}, M^{\veps j}\Big\}^{d-2}\ dudv
\nonumber
\end{eqnarray}
It suffices to consider the case $0\le v\le u$. Make the change of variables 
$x=u+v$, $y=u-v$. Then $x$ and $y$ are restricted to run over intervals
$J_1$, $J_2$ of length at most $4M^{\veps j}$ and
\begin{eqnarray}
&&\vol\set{\bk}{|e(\bk)|\le M^j, |\bk-\bq|\le M^{\veps j}}
\nonumber\\
&&\le
2c_{m,d}\int_{J_1\times J_2\atop 0\le y\le x}\hskip-5pt
    1\big(xy\le M^j\big)
        \min\big\{x, M^{\veps j}\big\}^{d-2}dxdy
\nonumber
\end{eqnarray}
In the event that $J_1\subset [M^{\veps j},\infty]$, then on the domain
of integration, $x\ge M^{\veps j}$ and the condition $xy\le M^j$ forces
$y\le M^j/M^{\veps j}$, so that
\begin{eqnarray}
&&\vol\set{\bk}{|e(\bk)|\le M^j, |\bk-\bq|\le M^{\veps j}}
\nonumber\\
&&\le
 2c_{m,d}\int_0^{M^{(1-\veps)j}}dy\int_{J_1} dx\ M^{(d-2)\veps j}
\nonumber\\
&&\le
2c_{m,d} M^{(1-\veps)j} 4M^{\veps j} M^{(d-2)\veps j}
\nonumber\\
&&= 8 c_{m,d} M^jM^{(d-2)\veps j}
\nonumber
\end{eqnarray}
If $J_1\cap [0,M^{\veps j}]\ne \emptyset$, the domain of integration is
contained in $0\le y\le x\le 5M^{\veps j}$ and 
\begin{eqnarray}
&&\vol\set{\bk}{|e(\bk)|\le M^j, |\bk-\bq|\le M^{\veps j}}
\nonumber\\
&&\le
2c_{m,d}\int_0^{5M^{\veps j}} dx\int_0^{5M^{\veps j}} dy\, 
        1\big( xy\le M^j\big)x^{d-2}
\nonumber
\end{eqnarray}
For $d=2$, the lemma follows from
$$
\int_0^{5M^{\veps j}}\hskip-15pt dx\int_0^{\min\{5M^{\veps j},M^j/x\}}\!\! dy
=
M^j\big(1+\ln 25+(1-2\veps)|j|\ln M\big)
$$
For $d>2$,
$$
\int_0^{5M^{\veps j}}\hskip-15pt dx\int_0^{\min\{5M^{\veps j},M^j/x\}}\!\! dy\ x^{d-2}
\le  
5^{d-1}M^jM^{(d-2)\veps j} .
$$
\end{proof}

\section{Improved power counting} 
\label{finer}

\null From now on we assume that $d \ge 3$, and that

\begin{itemize}

\item[$\circ$] $\cF=\set{\bk\in\bR^d}{e(\bk)=0}$ is compact

\item[$\circ$] $e(\bk)$ is $C^3$ 

\item[$\circ$] $\nabla e(\bk)$ vanishes only at isolated points of $\cF$.
We shall call them singular points.

\item[$\circ$] if $e(\bk)=0$ and $\bnabla e(\bk)=\bzer$, then
           $\big[\sfrac{\del^2\hfill}{\del\bk_i\del\bk_j} e(\bk)
                \big]_{1\le i,j\le d}$ is nonsingular and has at least 
                one positive eigenvalue and at least one negative eigenvalue.
\end{itemize}

In addition, we make an assumption that there is no nesting. 
In general, this means that any 
nontrivial translate of $\cF$ or $-\cF$ only has 
intersections with $\cF$ of at most some fixed finite degree.
Here we only require a weak form of no--nesting
-- namely that there is only polynomial flatness. 
This assumption, which is essentially the same as Hypothesis {\bf A3} 
in \cite{FST1}, is introduced and discussed in detail in the following.

\subsection{A no--nesting hypothesis and its consequences}
\label{sec:NN}
\noindent To make precise the ``only polynomial flatness''
hypotheses, let 
$$
n: \cF \to \bR^d,\qquad\om \mapsto n(\om)=\sfrac{\nabla e}{|\nabla e|} (\om)
$$
be the unit normal to the Fermi surface. It is defined except at singular 
points, which are isolated. For $\om,\om' \in \cF$, define the angle between 
$n(\om )$ and $n(\om')$ by
$$
\th (\om,{\om'} ) = \arccos ( n(\om) \cdot n({\om'}) )
$$
Let 
\begin{equation}
\cD (\om) = \set{ {\om'} \in \cF}{\abs{n({\om}) \cdot n(\om')} =1}
= \set{\om'\in \cF}{n(\om) =\pm n(\om')}
\label{cDomdef}
\end{equation}
and denote the $(d-1)$--dimensional measure of $A \subset \cF$ by $\vol_{d-1} A$.
Also, for any $A\subset\bR^d$ and $\beta> 0$ denote by
$
U_\beta (A) =\set{\bp\in\bR^d}{{\rm distance}(\bp,A)<\beta}
$
the open $\beta$-neighbourhood of $A$. We assume:

\bigskip\noindent
{\bf Hypothesis \hypflatnest. }{\em
There are strictly positive numbers $Z_0,\ \be_0$ and $\ka$
such that for all $\be \le \be_0$ and all $\om\in\cF$,
$$
\vol_{d-1}\set{\om'\in\cF}{ \abs{\sin \th (\om,{\om'})}=\sqrt{1-(n(\om')\cdot n(\om))^2} \le\be}
\le Z_0\be^\ka
$$
To verify this hypothesis, it suffices to find strictly positive numbers
$z_0, z_1, \rho', \beta_0 $ and $\ka'$ 
such that for all for all $\be \leq \be'_0$ and all $\om \in \cF$,
\begin{itemize}
\item[(i)]
$\vol_{d-1} \left( U_\be (\cD (\om)) \cap \cF \right) \leq z_0 \be^{\ka'} $
\item[(ii)] if ${\om'} \not\in U_\be (\cD (\om) ) \cap \cF$,
then $\abs{\sin \th (\om,{\om'})}= \sqrt{1-(n(\om)\cdot n({\om'}))^2} 
\geq z_1 \be^{\rh'}$.
\end{itemize}

\noindent Then $\ka=\sfrac{\ka'}{\rho'}$, $Z_0=z_0 z_1^{-\ka'/\rho'}$ and $\be_0=z_1{\be_0'}^{\rho'}$.

}

\bigskip\noindent
{\bf Example. } 
As an  example, take $d\ge 3$, $1\le m< d$ and
$e(\bk)=k_1^2+\ldots+k_m^2-k_{m+1}^2-\ldots-k_d^2$, say with $|\bk|\le\sqrt{2}$. 
The corresponding Fermi surface, $\cF$, is the (truncated) cone   
$k_1^2+\ldots+k_m^2=k_{m+1}^2+\ldots+k_d^2$, which we may parametrize by
$\bk=(r\bth,r\bphi)$ with $0\le r\le 1$, $\bth\in S^{m-1}$ and $\bphi\in
S^{d-m-1}$. The volume element on $\cF$ in this parametrization is 
$\sqrt{2}r^{d-2}\,dr\, d^{m-1}\bth\, d^{d-m-1}\bphi$, where $d^{m-1}\bth$ and  $d^{d-m-1}\bphi$
are the volume elements on $S^{m-1}$ and $S^{d-m-1}$ respectively.
The unit normals to $\cF$ at $\bk=(r\bth,r\bphi)$
are $\pm\sfrac{1}{\sqrt{2}}(\bth,-\bphi)$. 

Now fix any $\om=(r\bth,r\bphi)$ with $0<r\le 1$. Then
$$
D(\om)=\set{(t\bth,t\bphi)}{0<|t|\le 1}
$$
If $(t'\bth',t'\bphi')\in U_\be(D(\om))\cap\cF$ the there is a $t$ such
that
\begin{eqnarray}
\big|(t'\bth',t'\bphi')-(t\bth,t\bphi)\big|<\be
&\implies& \sqrt{{|t'\bth'-t\bth|}^2+{|t'\bphi'-t\bphi|}^2}<\be
\nonumber\\
&\implies& |t'\bth'-t\bth|<\be,\ |t'\bphi'-t\bphi|<\be,\ |t-t'|<\be
\nonumber\\
&\implies& |t'\bth'-t'\bth|<2\be,\ |t'\bphi'-t'\bphi|<2\be
\nonumber
\end{eqnarray} 
For each fixed $t'$ the volume of the $t'\bth'$s in $t'S^{m-1}$ for which
$|\bth'-\bth|<2\be/|t'|$ is at most a constant, depending only on $m$,
times $|t'|^{m-1}\min\big\{1,\big(\sfrac{\be}{|t'|}\big)^{m-1}\big\}
\le\be^{m-1}$ and the volume of the $t'\bphi'$s in $t'S^{d-m-1}$ for which
$|\bphi'-\bphi|<2\be/|t'|$ is at most a constant, depending only on $d-m-1$,
times $|t'|^{d-m-1}\min\big\{1,\big(\sfrac{\be}{|t'|}\big)^{d-m-1}\big\}
\le\be^{d-m-1}$.
Hence 
$$
\vol_{d-1}\big(U_\be(D(\om))\cap\cF\big)
\le c_{d,m}\int_0^1 dt'\ \be^{d-2}
= c_{d,m}\be^{d-2}
$$
Thus condition (i) of Hypothesis \hypflatnest\ is satisfied with $\ka'=d-2$. 

If $\om'=(t'\bth',t'\bphi')\notin U_\be(D(\om))\cap\cF$ then, for every 
$|t|\le 1$,
\begin{eqnarray}
\big|(t'\bth',t'\bphi')-(t\bth,t\bphi)\big|\ge \be
\nonumber
\end{eqnarray} 
In particular,
\begin{eqnarray}
\big|(t'\bth',t'\bphi')\pm(t'\bth,t'\bphi)\big|\ge \be
\implies \big|(\bth',\bphi')\pm(\bth,\bphi)\big|\ge \be
\nonumber
\end{eqnarray} 
The angle between
$
n(t'\bth',t'\bphi')=\pm\sfrac{1}{\sqrt{2}}(\bth',-\bphi')\hbox{ and }
n(r\bth,r\bphi)=\pm\sfrac{1}{\sqrt{2}}(\bth,-\bphi)
$
is the same ($\pm \pi$) as the angle between 
$(\bth',\bphi')$ and $(\bth,\bphi)$ (measured at the origin).
By picking signs appropriately, we may assume that $0\le \th\big(\om,\om')\le
\sfrac{\pi}{2}$. Thus
$$
\big|\sin\th\big(\om,\om')\big|
\ge\big|\sin\sfrac{1}{2}\th\big(\om,\om')\big|
=\sfrac{1}{2\sqrt{2}}
\big|(\bth',\bphi')\pm(\bth,\bphi)\big|\ge \sfrac{1}{2\sqrt{2}}\be
$$
and condition (ii) of Hypothesis \hypflatnest\ is satisfied with $\rho'=1$.
So $\ka=d-2$.

\begin{proposition}\label{propGenHyp}
Let $d\ge 3$ and let  $e:\bR^d\rightarrow\bR$ be $C^3$. Assume that
\begin{itemize}
\item[$\circ$] $e(\bzer )= 0$
\item[$\circ$] $\bnabla e(\bzer ) = \bzer $
\item[$\circ$] $\det\big[
                   \sfrac{\partial^2\hfill}{\partial\bk_i\partial\bk_j} e(\bzer )
                \big]_{1\le i,j\le d}\ne 0$
\item[$\circ$] $\big[
                   \sfrac{\partial^2\hfill}{\partial\bk_i\partial\bk_j} e(\bzer )
                \big]_{1\le i,j\le d}$ has $m\ge 1$ positive eigenvalues
and $d-m\ge 1$ negative eigenvalues.
\end{itemize}
\noindent Then there is a $c>0$ and constants $\be_0>0$ and $Z_0$
such that for every unit vector $\ba\in\bR^d$,
\begin{eqnarray}
&&\vol_{d-1}\set{\bk\in\cF}{ \sqrt{1-(n(\bk)\cdot \ba)^2} \le\be}
\le Z_0\be^{\max\{m-1,d-m-1\}}
\nonumber\\
&&\hbox{where }\cF=\set{\bk\in\bR^d}{|\bk|<c,\ e(\bk)=0}
\nonumber
\end{eqnarray}
for all $0<\be<\be_0$.
\end{proposition}

\begin{proof}
By a rotation, followed by a permutation of indices, we may assume 
that $\big[\sfrac{\partial^2\hfill}{\partial\bk_i\partial\bk_j} e(\bzer )
                \big]_{1\le i,j\le d}$ is a diagonal matrix, with diagonal
entries $2\la_1$, $2\la_2$, $\cdots$, $2\la_d$ that are in decreasing order.
By hypothesis, $\la_j>0$ for
$1\le j\le m$ and $\la_j<0$ for $m+1\le j\le d$. Replace $\la_j$ by $-\la_j$
for $j>m$. Then, 
$$
e(\bk)=\la_1k_1^2+\cdots+\la_mk_m^2-\la_{m+1}k_{m+1}^2-\cdots-\la_d k_d^2 +G(\bk)
$$
with $G(\bk)$ a $C^3$ function having a third order zero at $\bzer $. Define
\begin{eqnarray}
R_1(\bk)&=&\sqrt{\la_1k_1^2+\cdots+\la_mk_m^2}
\nonumber\\
R_2(\bk)&=&\sqrt{\la_{m+1}k_{m+1}^2+\cdots+\la_d k_d^2}
\nonumber\\
R(\bk)&=&\sqrt{\la_1k_1^2+\cdots+\la_d k_d^2}
\nonumber
\end{eqnarray}
Also use
\begin{eqnarray}
\tilde S_1^{m-1}&=&\set{(k_1,\cdots,k_m)}{\la_1k_1^2+\cdots+\la_mk_m^2=1}
\nonumber\\
\tilde S_2^{d-m-1}&=&\set{(k_{m+1},\cdots,k_d)}{\la_{m+1}k_{m+1}^2+\cdots+\la_d k_d^2=1}
\nonumber
\end{eqnarray}
to denote ``unit'' $(m-1)$--dimensional and $(d-m-1)$--dimensional
ellipsoids, respectively. 
For each $r>0$, the surface $R(\bk)=r$ is a $d-1$ dimensional ellipsoid
in $\bR^d$ with smallest semi--axis $r/\max_j\sqrt{\la_j}$ and largest
semi--axis  $r/\min_j\sqrt{\la_j}$. We now concentrate on the intersection 
of $\cF$ and that ellipsoid. The proof of Proposition \ref{propGenHyp} will
continue following the proof of Lemma \ref{lemAngle}.

\begin{lemma}\label{lemPertCone}
Suppose that
$$
|G(\bk)|\le g_0 R(\bk)^3\qquad |\nabla G(\bk)| \le g_1 R(\bk)^2
$$
and that $c$ is small enough (depending only on $g_0$, $g_1$ and the
$\la_i$'s). 
\begin{itemize}
\item[(a)] For each  $\bth_1\in \tilde S^{m-1}$, 
$\bth_2\in\tilde S^{d-m-1}$ and $r\ge 0$ 
such that the ellipsoid $\set{\bk\in\bR^d}{R(\bk)=r}$ is contained in the sphere
$\set{\bk}{|\bk|<c}$,
there is a unique $(r_1,r_2)$ such that
$$
r_1,r_2\ge 0\qquad r_1^2+r_2^2=r^2\qquad\hbox{and}\qquad (r_1\bth_1,r_2\bth_2) \in\cF
$$
Furthermore $|r_1-r_2|\le g_0 r^2$.
\item[(b)] $\cF$ is a $C^3$ manifold, except for a singularity at $\bk=0$.
\end{itemize}
\end{lemma}

\begin{proof}
(a)
 The point $(r_1\bth_1,r_2\bth_2)$ is on $\cF$ if and only if
$$
0=r_1^2-r_2^2+G(r_1\bth_1,r_2\bth_2)
=\big[r_1-r_2\big]\big[r_1+r_2\big]+G(r_1\bth_1,r_2\bth_2)
$$
or 
\begin{equation}
r_1-r_2=-\sfrac{G(r_1\bth_1,r_2\bth_2)}{r_1+r_2}
\label{eqnRay}
\end{equation}
For each $-r\le s\le r$ there are unique $r_1(s)\ge 0$ and $r_2(s)\ge 0$ such
that
$$
r_1(s)-r_2(s) = s\qquad r_1(s)^2+r_2(s)^2=r^2\qquad
{\figplace{sronertwo}{0.5 in}{-0.5 in}}
$$
Furthermore $r_1'(s)-r_2'(s)=1$ and $r_1(s)r_1'(s)+r_2(s)r_2'(s)=0$
gives that $r_1'(s)=\sfrac{r_2(s)}{r_1(s)+r_2(s)}$ and
$r_2'(s)=-\sfrac{r_1(s)}{r_1(s)+r_2(s)}$ have magnitude at most 1.
Since $r_1(s)+r_2(s)\ge r$, $H(s)=-\sfrac{G(r_1(s)\bth_1,r_2(s)\bth_2)}{r_1(s)+r_2(s)}$
obeys
$$
|H(s)|\le g_0 r^2
$$
and
\begin{eqnarray}
|H'(s)|&=&\Big|\sfrac{
    [r_1(s)+r_2(s)]\nabla G(r_1(s)\bth_1,r_2(s)\bth_2)
                                   \cdot(r'_1(s)\bth_1,r'_2(s)\bth_2)
      -[r_1'(s)+r_2'(s)] G(r_1(s)\bth_1,r_2(s)\bth_2)}
{[r_1(s)+r_2(s)]^2}\Big|
\nonumber\\
&\le&\sfrac{1}{r} g_1r^2|(r'_1(s)\bth_1,r'_2(s)\bth_2)|+\sfrac{1}{r^2}2g_0r^3
\nonumber\\
&\le& r\big[ g_1|(\bth_1,\bth_2)|+2g_0\big]
\nonumber\\
&\le& r\Big[ g_1\max_{1\le i\le d}\sqrt{\sfrac{2}{\la_i}}+2g_0\Big]
\nonumber\\
&<&1
\nonumber
\end{eqnarray}
provided $c$ is small enough. Consequently the function
$
s-H(s)
$
increases strictly monotonically from $-r -H(-r)\le -r+g_0r^2$ to 
$r-H(r)\ge r-g_0r^2$ as $s$ increases from $-r$ to $r$. So this function
has a unique zero and (\ref{eqnRay}) has a unique solution and the solution obeys
$|r_1-r_2|\le g_0r^2 $.

\Item{(b)}  Since
$$
\nabla e(\bk)=2(\la_1k_1,\cdots,\la_m k_m,-\la_{m+1}k_{m+1},\cdots,-\la_dk_d)
+\nabla G(\bk)
$$
and $|\nabla G(\bk)|\le g_1 R(\bk)^2\le g_1(\max_i\la_i)\ |\bk|^2$, the only zero of 
$\nabla e(\bk)$ is at $\bk=0$,
assuming that $c$ has been chosen small enough.
\end{proof}

\begin{lemma}\label{lemProjangle}
Define, for each $\ba,\bb\in\bR^d\setminus\{\bzer \}$,
$\th(\ba,\bb)\in [0,\pi]$ to be the angle between $\ba$ and $\bb$.  
Let $\rP_1:\bR^d\rightarrow\bR^m$ and 
$\rP_2:\bR^d\rightarrow\bR^{d-m}$ be the orthogonal projections onto
the first $m$ and last $d-m$ components of $\bR^d$, respectively.
Assume that the hypotheses of Lemma \ref{lemPertCone} are satisfied.
There is a constant $g_2$ (depending only on the $\la_i$'s)
such that if  $\bzer \ne \om\in\cF$, $\bzer \ne\ba\in\bR^d$ with
$|\sin\th\big(n(\om),\ba\big)|\le\be$, then
\begin{equation}
|\sin\th\big(\rP_1 n(\om),\rP_1\ba\big)|\le g_2\be\qquad
|\sin\th\big(\rP_2 n(\om),\rP_2\ba\big)|\le g_2\be
\label{eqnProj}
\end{equation}
\end{lemma}

\begin{proof}
We'll prove the first bound of (\ref{eqnProj}). We may assume that $\ba$ is a 
unit vector. 
Possibly replacing $\ba$ by $-\ba$, we may also assume that the angle between
$\ba$ and $n(\om)$ is at most $\sfrac{\pi}{2}$. By part (a)
of Lemma \ref{lemAngle}, below,
$$
|\ba-n(\om)|=2 \sin\sfrac{1}{2} \th(\ba,n(\om))\le 2 \sin\th(\ba,n(\om))\le 2\be
$$
So, by part (a) of Lemma \ref{lemAngle},
$$
\sin\th(\rP_1 \ba,\rP_1 n(\om))
\le \sfrac{|\rP_1 \ba-\rP_1 n(\om)|}{|\rP_1 n(\om)|}
\le \sfrac{|\ba-n(\om)|}{|\rP_1 n(\om)|}\le \sfrac{2\be}{|\rP_1 n(\om)|}
$$
Thus it suffices to prove that $|\rP_1 n(\om)|$ is bounded away from zero.
Recall that 
$$
\nabla e(\om)= n_1(\om)
+\nabla G(\om)
$$
where 
$$
n_1(\om)=2(\la_1k_1,\cdots,\la_m k_m,-\la_{m+1}k_{m+1},\cdots,-\la_dk_d)
$$
Use $\al\sim\ga$ to designate that there are constants $c,C>0$, depending
only on the $\la_i$'s, such that $c|\ga|\le|\al|\le C|\ga|$.
In this notation
$$
|n_1(\om)|\sim |\om|
\quad |\rP_1 n_1(\om)|\sim|\rP_1\om|\sim |R_1(\om)|
\quad |\rP_2 n_1(\om)|\sim|\rP_2\om|\sim |R_2(\om)|
$$
By part (a) of Lemma \ref{lemPertCone}, since $\om\in\cF$,
$$
|R_1(\om)- R_2(\om)|\le g_0 R(\om)^2 \qquad R_1(\om)^2+R_2(\om)^2=R(\om)^2
$$
As the maximum of $R_1(\om)$ and $R_2(\om)$ must be at least $\sfrac{1}{\sqrt{2}}
R(\om)$, we have
$$
R(\om)\ge R_1(\om),R_2(\om)\ge \sfrac{1}{\sqrt{2}}R(\om) - g_0 R(\om)^2\ge\sfrac{1}{2} R(\om)
$$
if $c$ is small enough. So
$$
|\rP_1n_1(\om)|,|\rP_2n_1(\om)|\sim |R(\om)|\sim |\om|
$$
As
$$
|\nabla G(\om)|\le g_1 R(\om)^2\le g_1(\max_i\la_i)\ |\om|^2
$$
we have that
$$
|\rP_1\nabla e(\om)|,|\rP_2\nabla e(\om)|\sim |\om|
$$
and hence that
$$
|\rP_1n(\om)|=\sfrac{|\rP_1\nabla e(\om)|}
            {\sqrt{|\rP_1\nabla e(\om)|^2+|\rP_2\nabla e(\om)|^2}}
$$
is bounded away from zero.
\end{proof}

\begin{lemma}\label{lemAngle}
Let $\ba,\bb\in\bR^d\setminus\{\bzer \}$.
\begin{itemize}   
\item[(a)] If $|\ba|=|\bb|$, then $\sin \sfrac{1}{2}\th(\ba,\bb)=\sfrac{1}{2}\sfrac{|\ba-\bb|}{|\ba|}$.
\item[(b)] For all $\ba,\bb\in\bR^d\setminus\{\bzer \}$,
$\sin\th(\ba,\bb)\le \sfrac{|\ba-\bb|}{|\ba|}$.
\end{itemize}
\end{lemma}

\begin{proof}
Part (a) is obvious from the figure on the left below. For part (b),
in the notation of the figure on the right below, we have, by the sin law
$$
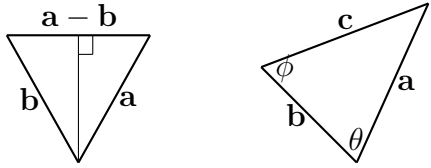
\sfrac{\sin\th}{|\bc|}=\sfrac{\sin\phi}{|\ba|}
\implies \sin\th =\sfrac{|\bc|}{|\ba|}\sin\phi\le \sfrac{|\bb-\ba|}{|\ba|}
$$

\centerline{\figput{angleA}\qquad\figput{angleB}}

\end{proof}

\proof{of Proposition \ref{propGenHyp} (continued).} 
Fix $\bk_2\in\bR^{d-m-1}$. If 
$\bk=(\bk_1,\bk_2)\in\cF$, then $\rP_1 n(\bk)$ is normal to $\cF_{\bk_2}=\set{\bk_1\in\bR^m}{(\bk_1,\bk_2)\in\cF}$
because both $n(\bk)$ and $\rP_2 n(\bk)$ are perpendicular to any 
vector $(\bt,\bzer )$ that is tangent to $\cF$ at $\bk$. The
matrix
$$
\big[\sfrac{\partial^2\hfill e\hfill\ }{\partial k_i\partial k_j}(\bk_1,\bk_2)
           \big]_{1\le i,j\le m}
=\big[2\la_i\de_{i,j}\big]_{1\le i,j\le m}
+\big[\sfrac{\partial^2\hfill G\hfill\ }{\partial k_i\partial k_j}(\bk_1,\bk_2)
           \big]_{1\le i,j\le m}
$$
is strictly positive definite (assuming that $c$ is small enough) because
$\sfrac{\partial^2\hfill G\hfill\ }{\partial k_i\partial k_j}(\bk)=O(|\bk|)$.
So the slice $\cF_{\bk_2}$ is strictly convex. The solution $(r_1,r_2)$
of Lemma \ref{lemPertCone} depends continuously on $\bth_1$, $\bth_2$ and
$r$, so, assuming that $m>1$, $\cF_{\bk_2}$ is connected.
Hence, for any fixed nonzero vector
$P_1\ba$, there are precisely two points of $\cF_{\bk_2}$ at which 
$|\sin\th\big(\rP_1 n(\bk_1,\bk_2),\rP_1\ba\big)|=0$. And at other points
$\bk_1\in\cF_{\bk_2}$, $|\sin\th\big(\rP_1 n(\bk_1,\bk_2),\rP_1\ba\big)|$ is
larger than a constant times the distance from $\bk_1$ to the nearest of
those two points. So
\begin{eqnarray}
&&\vol_{d-1}\set{\bk\in\cF}{ \big|\sin\th\big(n(\bk),\ba\big)\big| \le\be}
\nonumber\\
&&\hskip0.75in\le\const \sup_{\bk_2} \vol_{m-1}\set{\bk_1\in\cF_{\bk_2}}
               { \big|\sin\th\big(\rP_1 n(\bk),\rP_1\ba\big)\big|\le g_2\be}
\nonumber\\
&&\hskip0.75in\le\const \be^{m-1}
\nonumber
\end{eqnarray}
The bound
\begin{eqnarray}
\vol_{d-1}\set{\bk\in\cF}{ \big|\sin\th\big(n(\bk),\ba\big)\big| \le\be}
\le\const \be^{d-m-1}
\nonumber
\end{eqnarray}
is proven similarly.
\end{proof}
\begin{remark}\label{remNotOpt}The exponent $\ka=\max\{m-1,d-m-1\}$ of Proposition
\ref{propGenHyp} is not optimal, unless $m=1$ or $m=d-1$. Suppose that $2\le
m\le d-2$. As we observed in
the proof of Proposition \ref{propGenHyp}, for each fixed $\bk_2$ there are
precisely two distinct points of $\cF_{\bk_2}$ at which 
$\sin\th\big(\rP_1 n(\bk),\rP_1\ba\big)=0$. That is, at which $\rP_1 n(\bk)$
is parallel or antiparallel to $\rP_1\ba$. Hence
\begin{eqnarray}
&&\set{\bk\in\cF\setminus\{\bzer \}}{\sin\th\big(\rP_1 n(\bk),\rP_1\ba\big)=0} 
\nonumber\\
&&\hskip1in=\bigcup_{\bk_2\ne\bzer }\set{(\bk_1,\bk_2)}{\bk_1\in\cF_{\bk_2},\ \sin\th\big(\rP_1 n(\bk),\rP_1\ba\big)=0}
\nonumber
\end{eqnarray} 
consists of two disjoint $d-m$ dimensional submanifolds of $\cF$ and 
$$
\set{\bk\in\cF\setminus\{\bzer \}}{\big|\sin\th\big(\rP_1 n(\bk),\rP_1\ba\big)\big|<g_2\be} 
$$ 
consists of two tubes of thickness of order $\be$, and volume of order
$\be^{m-1}$, about those submanifolds. Similarly, 
$$
\set{\bk\in\cF\setminus\{\bzer \}}{\big|\sin\th\big(\rP_2 n(\bk),\rP_2\ba\big)\big|<g_2\be} 
$$ 
consists of two tubes of thickness of order $\be$, and volume of order
$\be^{d-m-1}$, about two disjoint $m$ dimensional submanifolds. In the ``free'' case, when $G=0$,
\begin{eqnarray}
\cF=\set{(r\bth_1,r\bth_2)}{|(r\bth_1,r\bth_2)|\le c,\ 
\bth_1\in\tilde S^{m-1}, \bth_2\in\tilde S^{d-m-1}}
\nonumber
\end{eqnarray}
and
$$
n\big(r\bth_1,r\bth_2)\parallel (\La_1\bth_1,-\La_2\bth_2)\quad\hbox{where }
\La_1=\big[\la_i\de_{i,j}\big]_{1\le i,j\le m},\ 
\La_2=\big[\la_i\de_{i,j}\big]_{m< i,j\le d}
$$
So
\begin{eqnarray}
\cM_1
&=&
\set{\bk\in\cF\setminus\{\bzer \}}{\sin\th\big(\rP_1 n(\bk),\rP_1\ba\big)=0}
\nonumber\\
&=&
\set{(r\bth_1,r\bth_2)}
     {0<|(r\bth_1,r\bth_2)|\le c,\ \bth_2\in\tilde S^{d-m-1},\ \bth_1\parallel\La_1^{-1}\rP_1\ba}
\nonumber\\
\cM_2
&=&
\set{\bk\in\cF\setminus\{\bzer \}}{\sin\th\big(\rP_2 n(\bk),\rP_2\ba\big)=0}
\nonumber\\
&=&
\set{(r\bth_1,r\bth_2)}
     {0<|(r\bth_1,r\bth_2)|\le c,\ \bth_1\in\tilde S^{m-1},\ \bth_2\parallel\La_2^{-1}\rP_2\ba}
\nonumber
\end{eqnarray}
intersect in the lines
$$
\cM_1\cap\cM_2
=\set{(r\bth_1,r\bth_2)}
     {0<|(r\bth_1,r\bth_2)|\le c,\ \bth_1\parallel\La_1^{-1}\rP_1\ba,\ \bth_2\parallel\La_2^{-1}\rP_2\ba}
$$
and otherwise cross transversely.
(If the $\la_i$'s are all the same, they cross perpendicularly.) 
So even when $G$ is nonzero, the tubes will cross
transversely (for sufficiently small $c$) and the volume of intersection 
will be of the order of the product $\be^{m-1}\be^{d-m-1}=\be^\ka$ with $\ka=d-2$.

\end{remark}

\subsection{The overlapping loop bound for $d\ge 3$}
In this section we prove the overlapping loop bound.
It generalizes the analogous bound of \cite[Proposition 1.1]{FST1} to singular Fermi surfaces
in $d \ge 3$.  The overlapping loop bound implies \cite{FST1}
that the first order derivatives of $\Si$ are bounded continuous functions
of momentum and frequency, to all orders in the renormalized expansion in the interaction,
and that the same holds for the counterterm function $K$. 

\begin{proposition}
\label{propOverlap}
Let $d \ge 3$, and let the dispersion relation $\bk \mapsto e(\bk)$ 
satisfy the generic assumptions stated at the beginning of Section \ref{finer} 
as well as the no--nesting hypothesis \hypflatnest.
Let $K,K_\bq$ be any compact subsets of 
$\bR^{2d}$ and $\bR$, respectively. There are constants $\veps>0$ and  
$\ \const\ $ such that for all $j_1,j_2,j_3<0$ and all $\bq\in K_\bq$,
\begin{eqnarray}
&&\vol\big\{(\bk,\bp)\in\bR^{2d}\cap K\,\big|\,|e(\bk)|\le M^{j_1},
|e(\bp)|\!\le\! M^{j_2},|e(\bq\pm\bk\pm\bp)|\!\le\! M^{j_3}\big\}
\nonumber\\
&&\le \const M^{j_{\pi(1)}}M^{j_{\pi(2)}}M^{\veps j_{\pi(3)}}
\nonumber
\end{eqnarray}
where $\pi$ is a permutation of $\{1,2,3\}$ with $j_{\pi(3)}=\max\{j_1,j_2,j_3\}$.
\end{proposition}

\begin{proof}
We may assume without loss of generality that $j_3=\max\{j_1,j_2,j_3\}$.
Otherwise make a change of variables with $\bk'= \bq\pm\bk\pm\bp,\ \bp'=\bk$
or $\bp$. By compactness, it suffices to show that for any $\tilde\bk,\ \tilde \bp$
and $\tilde \bq$ with $(\tilde\bk,\tilde\bp)\in K$ and $\tilde\bq\in K_\bq$,
there are constants $c$ and $\veps>0$ (possibly depending on $\tilde\bk$, 
$\tilde\bp$ and $\tilde\bq$, but independent of the $j_i$'s) such that
\begin{eqnarray}
\vol\Big\{\,(\bk,\bp)&\Big|& 
|e(\bk)|\le M^{j_1},|\bk-\tilde\bk|\le c,\ 
|e(\bp)|\le M^{j_2},|\bp-\tilde\bp|\le c,\ 
\nonumber\\
&&
\hskip0.6in|e(\bq\pm\bk\pm\bp)|\le M^{j_3} \Big\}
\le 
\const M^{j_1}M^{j_2}M^{\veps j_3}
\label{loverlapbnd}
\end{eqnarray}
for all $\bq$ with $|\bq-\tilde\bq|\le c$ and all $j_1,j_2,j_3<0$ with 
$j_3=\max\{j_1,j_2,j_3\}$. 

If any one of $e(\tilde\bk)$, $e(\tilde\bp)$, $e(\tilde\bq\pm\tilde\bk\pm\tilde\bp)$
is nonzero, the left hand side of \Ref{loverlapbnd} is exactly zero for 
all sufficiently small $c$ and sufficiently large $|j_3|$ (which also 
forces $|j_1|$ and $|j_2|$ to be sufficiently large). On the other hand,
for any bounded set of $j_3$'s, \Ref{loverlapbnd} follows from
\begin{eqnarray}
\vol\set{\bk\in\bR^d}{ |e(\bk)|\le M^{j_1},|\bk-\tilde\bk|\le c}
                     &\le&\const M^{j_1}
\nonumber\\
\vol\set{\bp\in\bR^d}{ |e(\bp)|\le M^{j_2},|\bp-\tilde\bp|\le c}
                     &\le&\const M^{j_2}
\nonumber
\end{eqnarray}
which holds by Lemma \ref{le:2.4}. So it suffices to consider 
$e(\tilde\bk)=e(\tilde\bp)=e(\tilde\bq\pm\tilde\bk\pm\tilde\bp)=0$.

By Lemma \ref{le:2.6}, if $\tilde\bk$ is a singular point,
then, for any $0\le \et<\sfrac{1}{2}$,
$$
\vol\set{\bk\in \bR^d}{|e(\bk)|\le M^j, 
|\bk-\tilde\bk|\le M^{\et j}, |\bk-\tilde\bk|\le c}
\le \const M^jM^{(d-2)\et j}
$$
Clearly, the same bound applies when $\tilde\bk$ is a regular point (that
is, if $\nabla e(\tilde\bk)\ne 0$).
 By replacing $(j,\et)$ with $\big(j_1,\sfrac{j_3}{j_1}\et\big)$ 
(observe that $\sfrac{j_3}{j_1}\et$ is still between $0$ and $\sfrac{1}{2}$), we have
\begin{eqnarray}
&&\vol\set{\bk\in \bR^d}{|e(\bk)|\le M^{j_1}, 
|\bk-\tilde\bk|\le M^{\et j_3}, |\bk-\tilde\bk|\le c}\nonumber\\
&&\hskip3.3in\le\const M^{j_1}M^{(d-2)\et j_3}
\nonumber
\end{eqnarray}
and hence
\begin{eqnarray}
&&\vol\set{\!\!(\bk,\bp)\!}{\!
|e(\bk)|\le M^{j_1},|\bk-\tilde\bk|\le M^{\et j_3}\!,
|e(\bp)|\le M^{j_2}\!,|\bp-\tilde\bp|\le c, 
\nonumber\\
&&\hskip0.83in
|e(\bq\pm\bk\pm\bp)|\le M^{j_3}\!\!\!}
\nonumber\\
&&\le \const M^{j_1}M^{(d-2)\et j_3}
     \  \vol\set{\bp\in\bR^d}{ |e(\bp)|\le M^{j_2},|\bp-\tilde\bp|\le c}
\nonumber\\
&&\le \const M^{j_1}M^{j_2}M^{(d-2)\et j_3}
\nonumber
\end{eqnarray}
Similarly,
\begin{eqnarray}
&&\vol\set{\!\!(\bk,\bp)\!}{\!
|e(\bk)|\le M^{j_1}\!,|\bk-\tilde\bk|\le c,
|e(\bp)|\le M^{j_2}\!,|\bp-\tilde\bp|\le M^{\et j_3}, 
\nonumber\\
&&\hskip0.83in
|e(\bq\pm\bk\pm\bp)|\le M^{j_3}\!\!\!}
\nonumber\\
&&\le \const M^{j_1}M^{j_2}M^{(d-2)\et j_3}
\nonumber
\end{eqnarray}
and
\begin{eqnarray}
&&\vol\Big\{\ (\bk,\bp)\ \Big|\ 
|e(\bk)|\le M^{j_1},|\bk-\tilde\bk|\le c,
|e(\bp)|\le M^{j_2},|\bp-\tilde\bp|\le c, 
\nonumber\\
&&\hskip1in
|e(\bq\pm\bk\pm\bp)|\le M^{j_3},
 |\bq\pm\bk\pm\bp-\tilde\bq\mp\tilde\bk\mp\tilde\bp|\le M^{\et j_3}\ \Big\}
\nonumber\\
&&\le \const M^{j_2}
     \ \sup_{\tilde\bk'}\vol\set{\bk}{ |e(\bk)|\le M^{j_1},|\bk-\tilde\bk|\le c,
                         |\bk-\tilde \bk'|\le M^{\et j_3}}
\nonumber\\
&&\le \const M^{j_1}M^{j_2}M^{(d-2)\et j_3}
\nonumber
\end{eqnarray}
Hence it suffices to prove that there is are $\tilde\veps>0$ and $0<\et<\sfrac{1}{2}$ 
such that
\begin{eqnarray}
&&\vol\Big\{\ (\bk,\bp)\ \Big|\ 
|e(\bk)|\le M^{j_1},M^{\et j_3}\le |\bk-\tilde\bk|\le c,\ 
\nonumber\\
&&\hskip0.95in
|e(\bp)|\le M^{j_2},M^{\et j_3}\le |\bp-\tilde\bp|\le c
\nonumber\\
&&\hskip0.75in |e(\bq\pm\bk\pm\bp)|\le M^{j_3},
     M^{\et j_3}\le |\bq\pm\bk\pm\bp-\tilde\bq\mp\tilde\bk\mp\tilde\bp|\le 3c \Big\}
\nonumber\\
&&\le \const M^{j_1}M^{j_2}M^{\tilde\veps j_3}
\label{loverlapbndext}
\end{eqnarray}
But, by hypothesis, 
$\big[\sfrac{\partial^2\hfill}{\partial\bk_i\partial\bk_j} e(\tilde\bk)
                \big]_{1\le i,j\le d}$ 
is nonsingular for every singular point $\tilde\bk$.
Hence, if $|\bk-\tilde\bk|\ge M^{\et j_3}$ for all singular points $\tilde\bk$, 
then $|\bnabla e(\bk)|\ge C M^{\et j_3}$ and if
$|\bp-\tilde\bp|\ge M^{\et j_3}$ for all singular points $\tilde\bp$, 
then $|\bnabla e(\bp)|\ge C M^{\et j_3}$ and if
$|\bq\pm\bk\pm\bp-\tilde\bq'|\ge M^{\et j_3}$ for all singular points $\tilde\bq'$, 
then $|\bnabla e(\bq\pm\bk\pm\bp)|\ge C M^{\et j_3}$. So, by 
Proposition \ref{NoNest} below, with $\de =CM^{\et j_3}$, $\veps_1=M^{j_1}$,
$\veps_2=M^{j_2}$ and $\veps_3=M^{j_3}$,
\begin{eqnarray}
&&\vol\Big\{\ (\bk,\bp)\ \Big|\ 
|e(\bk)|\le M^{j_1},M^{\et j_3}\le |\bk-\tilde\bk|\le c,
\nonumber\\
&&\qquad\qquad
|e(\bp)|\le M^{j_2},M^{\et j_3}\le |\bp-\tilde\bp|\le c
\nonumber\\
&& |e(\bq\pm\bk\pm\bp)|\le M^{j_3},
     M^{\et j_3}\le |\bq\pm\bk\pm\bp-\tilde\bq\mp\tilde\bk\mp\tilde\bp|\le 3c \Big\}
\nonumber\\
&&\le\const\sfrac{1}{\de^4} M^{j_1}M^{j_2}M^{\ep j_3}
\nonumber\\
&&=\const M^{j_1}M^{j_2}M^{(\ep-4\et) j_3}
\nonumber
\end{eqnarray}
If we choose $\et=\sfrac{\ep}{d+2}$, then 
$(d-2)\et=\ep-4\et=\sfrac{d-2}{d+2}\ep$ and the proposition follows
with
$$
\veps=\sfrac{d-2}{d+2}\ep=\sfrac{d-2}{d+2}\sfrac{\ka}{1+\ka}
$$
\end{proof}

We can now prove the volume improvement estimate that generalizes 
the one from \cite[Proposition 1.1]{FST1} to our situation. 

\begin{proposition}\label{NoNest} 
Let $K_\bk$, $K_\bp$ and $K_\bq$ be compact subsets of $\bR^d$ and 
$v_1,v_2\in\{+1,-1\}$. There are constants $C_{vol}$ and $C_\de$ such that
the following holds. Assume that there are $\de,\ka,\rho>0$ such that
\begin{itemize}

\item[(A1)] 
for all $\bk\in K_\bk$, $\bp\in K_\bp$ 
and $\bq\in K_\bq$:
$|\nabla e(\bk)|\ge\de$, $|\nabla e(\bp)|\ge \de$,  
and $|\nabla e(v_1\bk+v_2\bp+\bq)|\ge \de$ 

\item[(A2)] the ``only polynomial flatness''
                 condition of Hypothesis \hypflatnest\ is satisfied.
\end{itemize}

\noindent Set
\begin{equation}\label{EpsExp}
\ep = \sfrac{\ka}{1+\ka }
\end{equation}
Let
\begin{eqnarray}\label{Itwodef}
&&I_2 ( \veps_1, \veps_2, \veps_3 )
= \sup\limits_{\bq\in K_\bq}
\int_{K_\bk \times K_\bp} \hskip-20pt d^d \bk d^d \bp\ 
1\left(\abs{e(\bk)} \le \veps_1 \right)
1\left(\abs{e(\bp)} \le \veps_2 \right)
\nonumber\\
&&\hskip2.2in\times1\left(\abs{e(v_1\bk + v_2\bp + \bq)}\le\veps_3 \right)
\end{eqnarray}
Then, for all $0<\veps_1\le 1$, $0<\veps_2\le 1$, $\max\{\veps_1,\veps_2\}\le\veps_3\le 1$
with $\de\ge C_\de\max\{\sqrt{\veps_1},\sqrt{\veps_2}\}$
\begin{equation}\label{ImpVol}
I_2 ( \veps_1, \veps_2, \veps_3 ) \le C_{vol}\sfrac{1}{\de^4}
\veps_1 \veps_2 \veps_3^\ep .
\end{equation}
\end{proposition}

\begin{proof}
By compactness it suffices to assume that $K_\bk$ is contained either in
the ball $\set{\bk\in\bR^d}{|\bk-\tilde\bk|\le c}$ 
for some $\tilde\bk\in\cF$ with $\nabla e(\tilde\bk)\ne 0$ (i.e. $\tilde\bk$ is
a regular point)  or in the annulus
$\set{\bk\in\bR^d}{c'\de\le|\bk-\tilde\bk|\le c}$
for some $\tilde\bk\in\cF$ with $\nabla e(\tilde\bk)= 0$ (i.e. $\tilde\bk$ is a singular
point). We are free to choose $c,c'>0$, depending on $\tilde \bk$. 
We may make similar assumptions about $K_\bp$ and the allowed values of
$v_1\bk+v_2\bp+\bq$.

Make a change of variables from $\bk$ to $(\rho_1,\om_1)$, with $\rho_1=e(\bk)$.
We may assume that $K_\bk$ is covered by a single such 
coordinate patch, with Jacobian 
$$
|J_1(\rho_1,\om_1)|\le\sfrac{\const}{\de}
$$
In the case that $\tilde\bk$ is a singular point, we would use the Morse lemma, 
to provide a diffeomorphism $\bk(\bx)$ such that 
$$
e\big(\bk(\bx)\big)=x_1^2+\ldots+x_m^2-x_{m+1}^2-\ldots-x_d^2
$$
On the inverse image of $K_\bk$,
$$
2|\bx|=\big|\nabla_\bx e\big(\bk(\bx)\big)\big|
=\big|(\nabla_\bk e)\big(\bk(\bx)\big)^t\sfrac{\partial\bk}{\partial\bx}(\bx)\big|
\ge\const\de
$$
So we may first change variables from $\bk$ to $\bx$, with Jacobian bounded
and bounded away from zero (uniformly in $\de$) and then, in the region
where, for example $|x_1|\ge \const \max\{|x_2|,\ldots,|x_d|\}$, change
variables from $\bx$ to 
$$
(\rho,\om)=\big(x_1^2+\ldots+x_m^2-x_{m+1}^2-\ldots-x_d^2,x_2,\ldots,x_d\big)
$$
The second change of variables has Jacobian $2|x_1|\ge \const\de$. Observe
that, under this change of variables, the matrix 
$$
\sfrac{\partial\hfill\bk\hfill}{\partial(\rho,\om)}
=\sfrac{\partial\hfill\bk\hfill}{\partial\bx}
\left[\matrix{2x_1&2x_2&\ldots&-2x_d\cr
                0&&&\cr
                \vdots&&\bbbone&\cr
                0&&&\cr}\right]^{-1}
=\sfrac{\partial\hfill\bk\hfill}{\partial\bx}
\left[\matrix{\sfrac{1}{2x_1}&-\sfrac{x_2}{x_1}&\ldots&\sfrac{x_d}{x_1}\cr
                0&&&\cr
                \vdots&&\bbbone&\cr
                0&&&\cr}\right]
$$
has operator norm bounded by $\sfrac{\const}{\de}$. So $|\bk(\rho,\om)-\bk(0,\om)|
\le\const\sfrac{|\rho|}{\de}$.

Make a similar change of variables from $\bp$ to $(\rho_2,\om_2)$, with $\rho_2=e(\bp)$.
Again, we may assume that $K_\bp$ is covered by a single such 
coordinate patch, with Jacobian $|J_2(\rho_2,\om_2)|\le\sfrac{\const}{\de}$.
Then 
\begin{eqnarray}
I_2 (\veps_1, \veps_2, \veps_3) 
&\le&
\sup\limits_{\bq \in K_\bq} 
\int\limits_{-\veps_1}^{\veps_1} d \rh_1
\int\limits_{S_1} d\om_1\ J_1(\rh_1, \om_1 ) \int\limits_{-\veps_2}^{\veps_2} 
d \rh_2 \int\limits_{S_2} d\om_2\ J_2(\rh_2, \om_2 )
\nonumber\\
&&\hskip1in 1(\abs{e(v_1 \bk (\rh_1, \om_1 ) + 
v_2 \bp (\rh_2, \om_2 ) + \bq )} \leq \veps_3 ) 
\nonumber\\
& \le &
\const \veps_1 \veps_2\sfrac{1}{\de^2}
\sup\limits_{\bq \in K_\bq} 
\sup\limits_{\abs{\rh_1}, \abs{\rh_2} \leq \veps_3}
\int\limits_{S_1} d\om_1\int\limits_{S_2} d\om_2
\nonumber\\
&&\hskip1in1(\abs{e(v_1 \bk (\rh_1, \om_1 ) + 
v_2 \bp (\rh_2, \om_2 ) + \bq )} \leq \veps_3 )
\nonumber
\end{eqnarray}
By the mean value theorem
$$
\abs{e(v_1 \bk (\rh_1, \om_1 ) + v_2 \bp (\rh_2, \om_2 ) + \bq ) -
e\big(v_1 \bk (0 , \om_1 ) + v_2 \bp (0, \om_2 ) + \bq \big)} 
\le\! \const \sfrac{\veps_3}{\de}
$$
for all $\rh_1, \rh_2 $ with $\abs{\rh_i} \le \veps_3$.
Thus 
$$
\abs{e\big(v_1\bk (\rh_1, \om_1 ) + v_2 \bp (\rh_2, \om_2 ) + \bq )} \leq \veps_3
$$
implies 
$$
\abs{e(v_1 \bk(0 , \om_1 ) + v_2 \bp (0, \om_2 ) + \bq )} 
\le\const  \sfrac{\veps_3}{\de}
$$
and
$$
I_2 (\veps_1, \veps_2, \veps_3 ) \le
\const\veps_1 \veps_2\sfrac{1}{\de^2} W(\const\sfrac{\veps_3}{\de})
$$
with 
$$
W(\ze ) = 
\sup\limits_{\bq \in K_\bq} 
\int\limits_{S_1} d \om _1 \int\limits_{S_2} d \om _2 \ 
1(\abs{e(v_1 \bk (0 , \om_1 ) + v_2 \bp (0, \om_2 ) + \bq )} \le \ze )
$$

We claim that
$|\nabla e(v_1 \bk (0 , \om_1 ) + v_2 \bp (0, \om_2 ) + \bq )|\ge\const\de$
for all $\om_1\in S_1$ and $\om_2\in S_2$. This will be used in the proof
of the following Lemma, which generalizes \cite[Lemma A.1]{FST1} and
which implies the bound \Ref{ImpVol}. We have assumed
that $K_\bk$, $K_\bp$ and $K_\bq$ are contained in small balls or annuli
centred on $\tilde \bk$, $\tilde\bp$ and $\tilde\bq$ respectively. If
$v_1\tilde\bk+v_2\tilde\bp+\tilde\bq$ is a regular point, simple continuity
yields that 
$|\nabla e(v_1 \bk (0 , \om_1 ) + v_2 \bp (0, \om_2 ) + \bq )|\ge\const$
provided we chose $c$ small enough. So it suffices to consider the case
that $\br=v_1\tilde\bk+v_2\tilde\bp+\tilde\bq$ is a singular point.

The constraint $|\rho_1|<\veps_1$ ensures that 
$|\bk (\rho_1, \om_1 )-\bk (0 , \om_1 )|\le\const\sfrac{\veps_1}{\de}$
and the constraint $|\rho_2|<\veps_2$ ensures that 
$|\bk (\rho_2, \om_2 )-\bk (0 , \om_2 )|\le\const\sfrac{\veps_2}{\de}$.
So the original condition that 
$|\nabla e(v_1\bk(\rh_1,\om_1) + v_2\bp(\rh_2,\om_2) + \bq)|\ge\de$ implies
that 
$$
|v_1\bk(\rh_1,\om_1) + v_2\bp(\rh_2,\om_2)+ \bq-\br |\ge\const\de
$$
and hence 
$$
|v_1\bk (0, \om_1 ) + v_2 \bp (0, \om_2 ) + \bq-\br |\ge\const\de
-\sfrac{\veps_1}{\de}-\sfrac{\veps_2}{\de}\ge\const\de
$$
provided $\de\ge\const\max\{\sqrt{\veps_1},\sqrt{\veps_2}\}$. So
$$
|\nabla e(v_1 \bk (0 , \om_1 ) + v_2 \bp (0, \om_2 ) + \bq )|\ge\const\de
$$
as desired.
\end{proof}

\begin{lemma}\label{OnlyImp} 
$W(\ze) \le Z_3 \sfrac{1}{\de}\ze^\ep $ 
where $\ep=\sfrac{\ka}{1+\ka}$.
\end{lemma}

\begin{proof}
Let $\ga \in (0,1)$,  
$$
\cT = \Big \{\  (\om_1, \om_2 ) \in \cF \times \cF \ \Big|\ 
\sqrt{1-\left( n(\om_1 ) \cdot n(\om_2 )\right) ^2} \geq \ze^{1-\ga}\ \Big\}
$$
be the set where the intersection is transversal
and $\cE = \cF\times\cF\setminus \cT$ its complement.
We shall choose $\ga $ at the end.
Split $W(\ze ) = T(\ze ) + E(\ze )$ into the contributions from these two sets.

The contribution from the set of exceptional momenta $\cE $ is bounded 
using Hypothesis \hypflatnest. For each $\om_1 \in S_1$, let 
$$
\cE_{\om_1} =\Big \{\  \om_2 \in S_2\ \Big|\ 
\sqrt{1-\left( n(\om_1 ) \cdot n(\om_2 )\right) ^2} <
\ze^{1-\ga}\ \Big\}
$$
Then by Hypothesis \hypflatnest\ 
$$
E(\ze ) \le  \int\limits_{S_1} d\om_1 \int\limits_{\cE_{\om_1}\cap S_2} d\om_2 
\le  \int\limits_{S_1} d \om_1\ Z_0\ze^{\ka(1-\ga)}
=\const \ze^{\ka(1-\ga)}
$$

Now we bound $T$. We start by introducing a cover of $\cF$ by
coordinate patches. Let, for each singular point $\tilde\bk$ of $\cF$,
$\cO_{\tilde\bk}$ be the open neighbourhood of $\tilde\bk$ that is the image
of $\{|\bx|<1\}$ under the Morse diffeomorphism $\bk(\bx)$. 
If 
$$
e(\bk(\bx))=x_1^2+\ldots+x_m^2-x_{m+1}^2-\ldots-x_d^2
$$ 
write
$\bx=(r\bth_1,r\bth_2)$ with $0\le r\le 1/\sqrt{2}$, $\bth_1\in S^{m-1}$ and 
$\bth_2\in S^{d-m-1}$. Introduce ``roughly orthonormal'' coordinate patches
on $S^{m-1}$. 

Here is what we mean by the statement that $\bth_1(\al_1,\ldots,
\al_{m-1})$ is ``roughly orthonormal''. Let 
$$
A(\al_1,\ldots,\al_{m-1})=\Big[
\sfrac{\partial\hfill\bth_1\hfill\,}{\partial\al_1}(\al_1,\ldots,\al_{m-1}),\ldots,
\sfrac{\partial\hfill\bth_1\hfill\,}{\partial\al_{m-1}}(\al_1,\ldots,\al_{m-1})\Big]
$$
be the $m\times{m-1}$ matrix whose columns are the tangent vectors to the
coordinate axes at $\bth_1(\al_1,\ldots,\al_{m-1})$. The columns of this
matrix span the tangent space to $S^{m-1}$ at $\bth_1(\al_1,\ldots,\al_{m-1})$.
Let $V(\al_1,\ldots,\al_{m-1})$ be an $(m-1)\times(m-1)$ matrix such that
the columns of $A(\al_1,\ldots,\al_{m-1})V(\al_1,\ldots,\al_{m-1})$
are mutually orthogonal unit vectors. Those columns form an orthonormal basis 
for the tangent space to $S^{m-1}$ at $\bth_1(\al_1,\ldots,\al_{m-1})$. 
``Roughly orthogonal'' signifies
that $V$ and its inverse are uniformly bounded on the domain of the coordinate
patch. The only consequence of rough orthonormality that we will use is
that, if $v$ is any vector in the tangent space to $S^{m-1}$ 
at $\bth_1(\al_1,\ldots,\al_{m-1})$, then, because
\begin{eqnarray}
\|v\|
&=&
\Big\|v^t \big[ A(\al_1,\ldots,\al_{m-1})V(\al_1,\ldots,\al_{m-1})\big]\Big\|
\nonumber\\
&\le&\big\|v^t A(\al_1,\ldots,\al_{m-1})\big\|
                    \big\|V(\al_1,\ldots,\al_{m-1})\big\|
\nonumber
\end{eqnarray}
implies 
$$
\big\|v^t A(\al_1,\ldots,\al_{m-1})\big\|\ge \big\|V(\al_1,\ldots,\al_{m-1})\big\|^{-1} \|v\|
$$
we have
\begin{equation}
\max_{1\le j\le m-1}\big|v\cdot \sfrac{\partial\hfill\bth_1\hfill\,}{\partial\al_j}(\al_1,\ldots,\al_{m-1})\big|
\ge\sfrac{1}{\sqrt{m-1}}\big\|V(\al_1,\ldots,\al_{m-1})\big\|^{-1} \|v\|
\label{eqnRoughOrthA}
\end{equation}

Also introduce a ``roughly orthonormal'' coordinate patch
$\bth_2(\al_m,\ldots,\al_{d-2})$ on $S^{d-m-1}$ and parametrize (a patch
on) the cone $x_1^2+\ldots+x_m^2-x_{m+1}^2-\ldots-x_d^2=0$ by 
$$
\bx(\al_1,\ldots,\al_{d-1})
=\big(\al_{d-1}\bth_1(\al_1,\ldots,\al_{m-1}),
\al_{d-1}\bth_2(\al_m, \ldots, \al_{d-2})\big)
$$
and the corresponding patch on $\cO_{\tilde\bk}$ by
$\bk\big(\bx(\al_1,\ldots,\al_{d-1})\big)$. Denote
$$
\om_1(\al_1,\ldots,\al_{d-1})=\bk\big(\bx(\al_1,\ldots,\al_{d-1})\big)
$$ For  patches away from the
singular points, any roughly orthonormal coordinate systems will do.
Observe that, if $v$ is any vector in the tangent space to $\cF$
at $\om_1(\al_1,\ldots,\al_{d-1})$, then
\begin{equation}
\max_{1\le j\le d-1}\big|v\cdot \sfrac{\partial\hfill\om\hfill\,}{\partial\al_j}(\al_1,\ldots,\al_{d-1})\big|
\ge \|v\|\cases{\const & regular patch\cr\const \al_{d-1}& singular patch}
\label{eqnRoughOrthB}
\end{equation}

Now fix any $\bq\in K_\bq$ and consider the contribution to
$$
\dblInt_{S_1\times S_2\cap\cT}
d \om _1 d \om _2 \ 
1\Big(\abs{e(v_1 \bk (0 , \om_1 ) + v_2 \bp (0, \om_2 ) + \bq )} \le \ze \Big)
$$
from one pair, $\om_1(\al_1,\ldots,\al_{d-1})$ and
$\om_2(\be_1,\ldots,\be_{d-1})$, of coordinate patches as described above.
The Jacobian $\sfrac{\partial\hfill\om_1\hfill}{\partial\al_1\ldots\partial\al_{d-1}}$
is bounded by a constant, in the regular case, and a constant times
$\al_{d-1}^{d-2}$, in the singular case. Denote by 
$\th(\om_1,\om_2)$ the angle between $n(\om_1)$ and $n(\om_2)$.
By the transversality condition, $\sin\th(\om_1,\om_2)\ge \ze^{1-\ga}$. 
Consequently, for at least one $i\in\{1,2\}$ the sine of the 
angle between  $n(\om_i)$ and $\nabla e(v_1\bk(0,\om_1)+v_2\bp(0,\om_2)+\bq)$
 is at least 
$$
\sin\sfrac{1}{2}\th(\om_1,\om_2)\ge \sfrac{1}{2}\sin\th(\om_1,\om_2) \ge \sfrac{1}{2} \ze^{1-\ga}
$$ 
and 
the length of the
projection of $\nabla e(v_1\bk(0,\om_1)+v_2\bp(0,\om_2)+\bq)$ on $T_{\om_i }\cF$
must be at least $\sfrac{1}{2}\ze^{1-\ga}|\nabla e(v_1\bk(0,\om_1)+v_2\bp(0,\om_2)+\bq)|
\ge\const\de\,\ze^{1-\ga}$.
Suppose that $i=1$. Define
$$
\rho=e\big(v_1\bk(0,\om_1)+v_2\bp(0,\om_2)+\bq\big)
$$
viewed as a function of $\al_1,\ldots,\al_{d-1}$ and $\be_1,\ldots,\be_{d-1}$.
By \Ref{eqnRoughOrthB}, there must be a $1\le j\le d-1$ such that
\begin{eqnarray}
\big|\sfrac{\partial\hfill \rho\hfill}{\partial \al_j}\big|
&=&
\big|\nabla e(v_1\bk(0,\om_1)+v_2\bp(0,\om_2)+\bq)\cdot 
\sfrac{\partial\hfill\om\hfill\,}{\partial\al_j}(\al_1,\ldots,\al_{d-1})\big|
\nonumber\\
&\ge&
\const\de\, \ze^{1-\ga} \cases{\const & regular patch\cr
                          \const \al_{d-1}& singular patch}
\nonumber
\end{eqnarray} 
Make a final change of variables replacing $\al_j$ by $\rho$. The Jacobian 
for the composite change of variables from $(\om_1,\om_2)$ to
$(\al_1,\ldots,\al_{d-1},\be_1,\ldots,\be_{d-1})$ 
and then to
$\big((\al_i)_{1\le i\le d-1\atop i\ne j},
(\be_i)_{1\le i\le d-1},\rho)$ is bounded by 
$$
\const\sfrac{1}{\de}\ze^{\ga-1}\left.\cases{\const & regular patch\cr
   \const\, \al_{d-1}^{d-3}& singular patch}\right\}
\le \const\sfrac{1}{\de}\ze^{\ga-1}
$$
We thus have
$$
T(\ze ) \leq \const \sfrac{1}{\de} \ze^{\ga -1} 
\int_{-\ze}^{\ze} d\rho
\leq \const \sfrac{1}{\de}\ze^\ga
$$
The optimal bound is when $\ka (1-\ga ) =\ga $, that is,  
$\ga = \ka/ (1+\ka)$. 
\end{proof}

\subsection{The proof of Theorem \ref{thm:main}}\label{sec:mainproof}
\begin{mainproof}
Now that we have Proposition \ref{propOverlap}, the proof of 
Theorem \ref{thm:main} is almost identical to the corresponding
proofs of \cite{FST1}. The main change is that our current choice of 
localization operator simplifies the argument. Several proofs in this
paper and its companion paper \cite{SFS2} are variants of the arguments
of \cite{FST1}. So we have provided, in Appendix \ref{ap:A},
a complete, self--contained proof that the value, $G(q)$, of each renormalized
1PI, two--legged graph is $C^{1-\veps}$, using the simplest form of the
argument in question. In particular, it does not use ``volume improvement''
bounds like  Proposition \ref{propOverlap}. We here show how to use
Proposition \ref{propOverlap} to upgrade $C^{1-\veps}$ to $C^{1+\veps}$.
This is a good time to read that Appendix, since we shall just explain 
the modifications to be made to it.

As in Appendix \ref{ap:A}, use \Ref{eq:Cexpn} to introduce a 
scale expansion for each propagator and express $G(q)$ in terms of a 
renormalized tree expansion \Ref{eq:Gren}. We shall prove, 
by induction on the depth, $D$, of $G^J$, the bound
\begin{equation}\label{eq:SFS1derivindhyp}
\sum_{J\in \cJ(j,t,R,G)}\hskip-15pt
\sup\nolimits_{q}\big|\partial_{q_0}^{s_0}\partial_\bq^{s_1}G^J(q)\big|
\le\cst{}{n} |j|^{3n-2}M^{(1-s_0-s_1)j}
\cases{M^{\veps j}& if $s_0+s_1\ge 1$\cr 1&if $s_0=s_1=0$\cr}
\end{equation}
for $s_0,s_1\in\{0,1,2\}$.  Here $\veps$ was specified in Proposition 
\ref{propOverlap} and the other notation is as in Appendix \ref{ap:A}: 
$n$ is the number of vertices in $G$ and $\cJ(j,t,R,G)$ is the set of all 
assignments $J$ of scales to the lines of  $G$ that have root scale $j$, that 
give forest $t$ and that are compatible with the assignment $R$ of 
renormalization labels to the two--legged forks of $t$. (This is explained in 
more detail just before \Ref{eq:Gren}.) If $s_0+s_1=1$, 
the right hand side becomes $\cst{}{n}|j|^{3n-2}M^{\veps j}$, which is summable 
over $j<0$, implying that 
$G(q)=\sum_{R}\sum_{j<0}\sum_{J\in \cJ(j,t,R,G)}G^J(q)$ is $C^1$. 
To show that the first order derivatives of $G(q)$ are H\"older continuous
of any degree strictly less than $\veps$, just observe that if
$$
\|f_j\|_\infty\le \cst{}{n} |j|^{3n-2}M^{\veps j}\qquad\hbox{and}\qquad
\|f'_j\|_\infty\le \cst{}{n} |j|^{3n-2}M^{\veps j}M^{-j}
$$
then
\begin{eqnarray*}
\big|f_j(x)-f_j(y)\big|
&\le&\min\big\{2\|f_j\|_\infty,\|f'_j\|_\infty|x-y|\big\}\\
&\le& \cst{}{n} |j|^{3n-2}M^{\veps j}\min\big\{2,M^{-j}|x-y|\big\}\\
&\le& \cst{}{n} |j|^{3n-2}M^{\veps j}M^{-\et j}|x-y|^\et\qquad\hbox{ for
any }0\le\et\le 1
\end{eqnarray*}
is summable over $j<0$ for any $0<\et<\veps$.

If $s_0=s_1=0$, \Ref{eq:SFS1derivindhyp} is contained in  
Proposition \ref{le:fixedrootscale}, so it suffices to 
consider $s_0+s_1\ge 1$.
As in Appendix \ref{ap:A}, if $D>0$, decompose the tree $t$ into a 
pruned tree $\tilde t$ and insertion subtrees $\tau^1,\cdots,\tau^m$ by cutting
the branches beneath  all minimal $E_f=2$ forks $f_1,\cdots,f_m$. In other
words each of the forks $f_1,\cdots,f_m$ is an $E_f=2$ fork having
no $E_f=2$ forks, except $\phi$, below it in $t$. Each $\tau_i$
consists of the fork $f_i$ and all of $t$ that is above $f_i$. It has depth
at most $D-1$ so the corresponding subgraph $G_{f_i}$ obeys
\Ref{eq:SFS1derivindhyp}. Think of each subgraph $G_{f_i}$ as
a generalized vertex in the graph $\tilde G=G/\{G_{f_1},\cdots,G_{f_m}\}$.
Thus $\tilde G$ now has two as well as four--legged vertices. These
two--legged vertices have kernels of the form
$
T_i(k)=\sum_{j_{f_i}\le j_{\pi(f_i)}}\ell G_{f_i}(k)
$
when $f_i$ is a $c$--fork and of the form
$
T_i(k)=\sum_{j_{f_i}> j_{\pi(f_i)}}(\bbbone-\ell)G_{f_i}(k)
$
when $f_i$ is an $r$--fork. At least one of the external 
lines\footnote{Note that the root fork, $\emptyset$, of 
\Ref{eq:Gren} does not carry an $r,c$ label so that 
$\tilde G$ may not be simply a single two--legged $c$-- or $r$--vertex.
At least one external line of each $G_{f_i}$ must be an internal line
of $\tilde G$.} 
of $G_{f_i}$ must be of scale precisely $j_{\pi(f_i)}$ so
the momentum $k$ passing through $G_{f_i}$ lies in the support of 
$C_{j_{\pi(f_i)}}$. In the case of a $c$--fork $f=f_i$ we have, as in
\Ref{eq:cforkA} and using the same notation, by the 
inductive hypothesis,
\begin{eqnarray}
&&\hskip-23pt\sum_{j_{f}\le j_{\pi(f)}}\sum_{J_f\in\cJ(j_f,t_f,R_f,G_f)}
\hskip-6pt\sup_{k}\Big|\partial_\bk^{s'_1}\ell G_{f}^{J_f}(k)\Big|
\le\sum_{j_{f}\le j_{\pi(f)}}\hskip-7pt\cst{}{n_f}|j_f|^{3n_f-2}M^{j_f}
         M^{-s'_1(1-\veps)j_f}\nonumber\\
&&\hskip70pt
\le \cst{}{n_f}|j_{\pi(f)}|^{3n_f-2}M^{j_{\pi(f)}}M^{-s'_1(1-\veps)j_{\pi(f)}}
\label{eq:cforkSFS1} 
\end{eqnarray}
for $s'_1=0,1$. Note that the sum in the analog of \Ref{eq:cforkSFS1}
diverges when $s'_1=2$, so it is essential that no more than one derivative
act on any $c$--fork.
As $\ell G_{f}^{J_f}(k)$ is independent of $k_0$, derivatives with respect
to $k_0$ may not act on it. In the case of an $r$--fork $f=f_i$, we have, as in 
\Ref{eq:rforkA}, using the mean value theorem in the
case $s'_0=0$,
\begin{eqnarray}
&&\hskip-20pt\sum_{j_{f}> j_{\pi(f)}}\sum_{J_f\in\cJ(j_f,t_f,R_f,G_f)}
\sup_{k}1\big(C_{j_{\pi(f)}}(k)\ne 0\big)
\Big|\partial_{k_0}^{s'_0}\partial_\bk^{s'_1}
(\bbbone-\ell)G_{f}^{J_f}(k)\Big|\nonumber\\
&&\hskip20pt
\le \sum_{j_{f}> j_{\pi(f)}}\ \sum_{J_f\in\cJ(j_f,t_f,R_f,G_f)}
M^{(1-\min\{1,s'_0\})j_{\pi(f)}}
\sup_{k}\Big|\partial^{\max\{1,s'_0\}}_{k_0}\partial_\bk^{s'_1}
    G_{f}^{J_f}(k)\Big|\nonumber\\
&&\hskip20pt
\le\cst{}{n_f}M^{(1-\min\{1,s'_0\})j_{\pi(f)}}
\sum_{j_{f}> j_{\pi(f)}}|j_{f}|^{3n_f-2}
     M^{-(\max\{1,s'_0\}+s'_1-1-\veps)j_f}\nonumber\\
&&\hskip20pt
\le\cst{}{n_f}|j_{\pi(f)}|^{3n_f-2}M^{j_{\pi(f)}}M^{-s'_0j_{\pi(f)}}
M^{-s'_1j_{\pi(f)}}
\label{eq:rforkSFS1}
\end{eqnarray}

Denote by $\tilde J$ the restriction to $\tilde G$ of the scale assignment
$J$. We bound $\tilde G^{\tilde J}$, which again is of the form 
\Ref{eq:tildeGtildeJform}, by a variant of the six step 
procedure followed in Appendix \ref{ap:A}. In fact the first five steps 
are almost identical.
\begin{enumerate} 
\item
Choose a spanning tree $\tilde T$ for $\tilde G$ with the property that 
$\tilde T\cap \tilde G^{\tilde J}_f$ is a connected tree for every 
$f\in t(\tilde G^{\tilde J})$. 
\item
 Apply any $q$--derivatives. By the product rule each derivative may act on any 
line or vertex on the ``external momentum path''. It suffices to consider 
any one such action. Ensure, through a judicious use of integration by
parts, that at most one derivative acts on any single $c$--fork. To do so, 
observe that a derivative with respect to the external momentum acting on 
a $c$--fork is, up to a sign, equal to the derivative with respect to any loop 
momentum that flows through the fork. So replace one external momentum
derivative by a loop momentum derivative and integrate by parts to move
the latter off of the $c$--fork.
\item
 Bound each two--legged renormalized subgraph (i.e.
$r$--fork) by \Ref{eq:rforkSFS1} and each two--legged 
counterterm (i.e. $c$--fork) by \Ref{eq:cforkSFS1}. 
Observe that when $s'_0$ $k_0$--derivatives and $s'_1$ $\bk$--derivatives 
act on the vertex, the bound is no worse than 
$M^{-s'_0j_{\pi(f)}}M^{-s_1'j_{\pi(f)}}$ times the bound with no derivatives. 
(We shall not need the factor $M^{s_1'\veps j_{\pi(f)}}$
in \Ref{eq:cforkSFS1}. So we simply discard it.)
As we have already observed, one of the external lines of the two--legged
vertex must be of scale precisely $j_{\pi(f)}$. We write 
$M^{-s'_0j_{\pi(f)}}M^{-s_1'j_{\pi(f)}}=M^{-s'_0j_\ell}M^{-s_1'j_\ell}$,
where $\ell$ is that line.
\item
 Bound all of the remaining vertex functions, (suitably differentiated)
by their suprema in momentum space. We have already observed that if 
$s_0=s_1=0$, the target bound \Ref{eq:SFS1derivindhyp} is contained in  
Proposition \ref{le:fixedrootscale}, with $s_0=s=0$.
In the event that $s_0+s_1\ge 1$, but all derivatives act on four--legged
vertex functions, Proposition \ref{le:fixedrootscale}, again 
with $s_0=s=0$ but with one or two four--legged vertex functions replaced 
by differentiated functions, again gives \Ref{eq:SFS1derivindhyp}.
So it suffices to consider the case that at least one derivative acts on
a propagator or on a $c$-- or $r$--fork. 
\item
 Bound each  propagator 
\begin{eqnarray}\label{eq:propbndSFS1}
&&\hskip-25pt\big|\partial^{s'_0}_{k_0}\partial^{s'_1}_{\bk}C_{j_\ell}(k)\big|
\le\const M^{-(1+s'_0+s'_1)j_\ell}\ 1\big(|ik_0-e(\bk)|\le M^{j_\ell}\big)
\nonumber\\
\end{eqnarray}
Once again,  when $s'_0$ $k_0$--derivatives and $s'_1$ $\bk$--derivatives act 
on the propagator, the bound is no worse than $M^{-(s'_0+s')j_\ell}$ times 
the bound with no derivatives. 
\end{enumerate}

 We now have 
$\big|\partial_{q_0}^{s_0}\partial_{\bq}^{s_1}\tilde G^{\tilde J}(q)\big|$
bounded, uniformly in $q$ by
\begin{eqnarray}\label{eq:step5bnd}
&&\hskip-20pt\cst{}{n}\prod_{d}M^{-j_{\ell_d}}
   \prod_{i=1}^m|j_{\pi(f_i)}|^{3n_{f_i}-2}M^{j_{\pi(f_i)}}
    \prod_{\ell\in \tilde G}M^{-j_\ell} \nonumber\\
&&\hskip1in\int\prod_{\ell\in \tilde G\setminus \tilde T}\dbar k_\ell
\ \prod_{\ell\in \tilde G}1\big(
|ik_{\ell0}-e(\bk_\ell)|\le  M^{j_\ell}\big)
\end{eqnarray}
Here $d$ runs over the $s_0+s_1\ge 1$ derivatives in 
$\partial_{q_0}^{s_0}\partial_{\bq}^{s_1}$ and $\ell_d$ refers to the 
specific line on which the derivative acted (or, in the case that the 
derivative acted on a $c$-- or $r$--fork, the external line specified in 
step 3). 

For $\ell\in\tilde T$, the momentum $k_\ell$ is a signed sum of the loop
momenta and external momentum flowing through $\ell$.
In Appendix \ref{ap:A}, we discarded the factors of the integrand
$\prod_{\ell\in \tilde G}
1\big(|ik_{\ell0}-e(\bk_\ell)|\le  M^{j_\ell}\big)$ with $\ell\in\tilde T$
at this point. Then the integrals over the loop momenta factorized and
we bounded them by the volumes of their domains of integration, using 
Lemma \ref{le:2.4}. We now deviate from the argument of Appendix \ref{ap:A} by exploiting the constraint that one factor
$1\big(|ik_{\ell0}-e(\bk_\ell)|\le  M^{j_\ell}\big)$ with $\ell\in\tilde T$
imposes on the domain of integration.

We have reduced consideration to cases in which at least one derivative
with respect to the external momentum acts either on a propagator in $\tilde T$ 
or on a two--legged $c$-- or $r$--vertex in $\tilde T$, so that 
the associated line $\ell_d\in\tilde T$.  Select any such $\ell_{d_0}$.
Recall that any line $\ell\in\tilde G\setminus\tilde T$ is associated with a 
loop $\La_\ell$ that consists of $\ell$ and the linear subtree of $\tilde T$ 
joining the vertices at the ends of $\ell$. 
By \cite[Lemma \ref{FS2-le:overlap}]{SFS2}, there exist two lines 
$\ell_1,\ell_2\in\tilde G\setminus\tilde T$ such that 
$\ell_{d_0}\in\La_{\ell_1}\cap\La_{\ell_2}$.
By \cite[Lemma \ref{FS2-le:overlapScale}]{SFS2}, $j_{\ell_1},j_{\ell_2}\le
j_{\ell_{d_0}}$. Now discard all of the factors $\prod_{\ell\in \tilde G}
1\big(|ik_{\ell0}-e(\bk_\ell)|\le  M^{j_\ell}\big)$ in the integrand of 
\Ref{eq:step5bnd} with $\ell\in\tilde T\setminus\{\ell_{d_0}\}$. Choose the
order of integration in \Ref{eq:step5bnd} so that $k_{\ell_1}$ and $k_{\ell_2}$
are integrated first. By Proposition \ref{propOverlap}, 
\begin{eqnarray}\label{eq:overlaploopbnd}
\int\prod_{\ell\in \{\ell_1,\ell_2\}}\hskip-8pt\dbar k_\ell
\ \prod_{\ell\in \{\ell_1,\ell_2,\ell_{d_0}\}}\hskip-15pt
1\big(|ik_{\ell0}-e(\bk_\ell)|\le  M^{j_\ell}\big)
\le\const M^{2j_{\ell_1}}M^{2j_{\ell_2}}M^{\veps j_{d_0}}
\end{eqnarray}
Finally, integrate over the remaining loop momenta $k_\ell$, $\ell\in\tilde G
\setminus(\tilde T\cup\{\ell_1,\ell_2\})$ as in step 6 of 
\Ref{le:fixedrootscale}. The integral over each such $k_\ell$ is 
bounded by vol$\set{k_{0\ell}}{|k_{0\ell}|\le M^{j_\ell}}\le 2M^{j_\ell}$ 
times the volume of $\set{\bk_\ell}{|e(\bk_\ell)|\le M^{j_\ell}}$, which, by
Lemma \ref{le:2.4}, is bounded by a constant times $|j_\ell|M^{j_\ell}$.
 We now have 
$\big|\partial_{q_0}^{s_0}\partial_{\bq}^{s_1}\tilde G^{\tilde J}(q)\big|$
bounded, uniformly in $q$, by
\begin{eqnarray*}
&&\hskip-40pt\cst{}{n}M^{\veps j_{\ell_{d_0}}} \prod_{d}M^{-j_{\ell_d}}
   \prod_{i=1}^m|j_{\pi(f_i)}|^{3n_{f_i}-2}M^{j_{\pi(f_i)}}
    \prod_{\ell\in \tilde G}M^{-j_\ell} 
    \prod_{\ell\in \tilde G\setminus  \tilde T}|j_\ell|  M^{2j_\ell}
\end{eqnarray*}
For every derivative $d$, $j_{\ell_d}\ge j=j_\phi$,
so that
$$
M^{\veps j_{\ell_{d_0}}} \prod_{d}M^{-j_{\ell_d}}
=M^{-(1-\veps) j_{\ell_{d_0}}} \prod_{d\ne d_0}M^{-j_{\ell_d}}
\le M^{-s_0j-s_1j}M^{\veps j}
$$
Bounding each $|j_{\pi(f_i)}|^{3n_{f_i}-2}\le |j_{\pi(f_i)}|^{3n_{f_i}-1}$, 
we come to the conclusion that 
$\big|\partial_{q_0}^{s_0}\partial_{\bq}^{s_1}\tilde G^{\tilde J}(q)\big|$
is bounded, uniformly in $q$, by
\begin{eqnarray}\label{eq:step6bndA}
&&\hskip-40pt\cst{}{n}M^{-s_0j-s_1j}M^{\veps j}
   \prod_{i=1}^m|j_{\pi(f_i)}|^{3n_{f_i}-1}M^{j_{\pi(f_i)}}
    \prod_{\ell\in \tilde G}M^{-j_\ell} 
    \prod_{\ell\in \tilde G\setminus  \tilde T}|j_\ell|  M^{2j_\ell}
\end{eqnarray}
This is exactly $M^{-s_0j-s_1j}M^{\veps j}$ times the bound
\Ref{eq:step6bnd}$\big|_{s_0=s=0}$ of Appendix \ref{ap:A}. So
\Ref{eq:scalesumbnd}$\big|_{s_0=s=0}$ of Appendix \ref{ap:A} now gives
\Ref{eq:SFS1derivindhyp}. This completes the proof that the value of each
graph contributing to the self--energy is $C^{1+\et}$ in the external momentum,
for every $\et$ strictly less than the $\veps$ of Proposition \ref{propOverlap}.

We may also apply this technique to connected four--legged graphs.
There is no need for an induction argument because we already have all
of the needed bounds on $c$-- and $r$--forks. We just need to go through
the $\tilde G$ argument once. When we do so, there are three changes:
\begin{itemize}
\item 
The overall power counting factor $M^{{1\over 2}j(4-E_\phi)}$
in \Ref{eq:pwrcntbnd}, which took the value $M^j$ for 
two--legged graphs, now takes the value $1$ for four--legged graphs.
\item 
We may only apply the overlapping loop bound \Ref{eq:overlaploopbnd}
when we can find a line $\ell_3\in\tilde T$  and two lines 
$\ell_1,\ell_2\in\tilde G\setminus\tilde T$ with 
$\ell_3\in\La_{\ell_1}\cap\La_{\ell_2}$.  By \cite[Lemma 2.34]{FST1}, this 
is the case if and only if $\tilde G$ is overlapping, as defined in 
\cite[Definition 2.19]{FST1}. By \cite[Lemma 2.26]{FST1}, four--legged
connected graphs fail to be overlapping if and only if they are dressed bubble 
chains, as defined in \cite[Definition 2.24]{FST1}.
\item 
To convert the $M^{\veps j_{\ell_3}}$ from the overlapping loop integral
\Ref{eq:overlaploopbnd} into the $M^{\veps j}$ that we want in the final 
bound, we set $f_3$ to the highest fork with $\ell_3\in \tilde G_{f_3}$
and write
$$
M^{\veps j_{\ell_3}}
=M^{\veps j_{f_3}}
=M^{\veps j}\prod_{f\in \tilde t\atop \phi<f\le f_3}
M^{\veps(j_f-j_{\pi(f)})}
$$
The extra factors $M^{\veps(j_f-j_{\pi(f)})}$, which are all at least one, 
are easily absorbed by \Ref{eq:scalesum} provided 
$E_f>4$ for all forks $f$ between the root $\phi$ and $f_3$. We may 
choose $\ell_3$ so that this is the case precisely when $G$ is {\it not} 
a generalized ladder. To see this, let $\hat G^J=G^J/\big\{G^J_f\ |\ E_f=2,4\}$ 
be the diagram $G$, but with both two-- and four--legged subdiagrams $G_f$
viewed as generalized vertices. Then we can find a suitable $\ell_3$ if
and only if $\hat G^J$ is overlapping which in turn is the case if and
only if $\hat G^J$ is not a dressed bubble chain, which in turn is
the case, for all labellings $J$, if and only if $G$ is not a generalized
ladder. 
\end{itemize}
Thus when $G$ is a connected four--legged graph, the right hand side of 
\Ref{eq:SFS1derivindhyp} is replaced by a constant times a power of $j$
times
$$
M^{(-s_0-s_1)j}
\cases{M^{\veps j}& if $G$ is not a generalized ladder\cr 1& if
$G$ is a generalized ladder\cr}
$$
for $s_0,s_1\in\{0,1\}$.
This implies that four--legged graphs, other than generalized ladders, 
are $C^\et$ functions of their external momenta for all $\et$ strictly smaller 
than $\veps$. 

For graphs $G$ contributing to the higher correlation functions, we may once again
repeat the same argument, but with $s_0=s_1=0$ and without having to 
exploit overlapping loops, provided we use the $L^1$ norm, rather than
the $L^\infty$ norm, on the momentum space kernel of $G$. In \cite{FST1},
this norm was denoted $|\ \cdot\ |'$ and was defined in (1.46).  
See \cite[(2.27) and Theorem 2.47]{FST1} for the proof.

Denote by $K(e,\bq)$ the counterterm function for the dispersion relation
$e(\bk)$ and by 
$$
C_j(e,k)=\sfrac{f(M^{-2j}|ik_0-e(\bk)|^2)}{ik_0-e(\bk)}
$$
the scale $j$ propagator for the dispersion relation $e(\bk)$. Observe
that, for all $j_\ell<0$ and $s'_0\in\{0,1\}$,
$$
\big|\sfrac{\partial\hfill}{\partial t}\partial^{s'_0}_{k_0}
   C_{j_\ell}(e+th,k)\big|_{t=0}\big|
\le\const \|h\|_\infty
M^{-(2+s'_0)j_\ell}\ 1\big(|ik_0-e(\bk)|\le M^{j_\ell}\big)
$$
Thus the effect of a directional derivative with respect to the dispersion
relation in direction $h$ is to multiply \Ref{eq:propbndSFS1} by
$\|h\|_\infty M^{-j_\ell}$, which is $\|h\|_\infty$ times the effect of
a $\partial_{k_0}$ derivative. So the same argument that led to 
\Ref{eq:SFS1derivindhyp} gives
$$
\sum_{J\in \cJ(j,t,R,G)}\hskip-15pt
\sup\nolimits_{q}
\big|\sfrac{\partial\hfill}{\partial t}\partial_{q_0}^{s_0}G^J(e+th,q)\big|_{t=0}\big|
\le\cst{}{n} |j|^{3n-2}M^{-s_0 j}M^{\veps j}\|h\|_\infty
$$
for $s_0\in\{0,1\}$. When $s_0=0$, this is summable over $j<0$ so that
\begin{equation}\label{eq:dirderiv}
\sup_{\bq}\Big|
\sfrac{\partial\hfill}{\partial t}\Big[\sum_{r=1}^R\la^rK_r(e+th,\bq)\big|_{t=0}
\Big]\Big|
\le\cst{}{{\rm dKde}}\ |\la|\|h\|_\infty
\end{equation} 
The constant $\cst{}{{\rm dKde}}=\cst{}{{\rm dKde}}(e,v)$ depends on $R$
and the
various parameters in the hypotheses imposed by Theorem \ref{thm:main}
on the dispersion relation $e$ and two--body interaction $v$, like
the $C^3$ norm of $e$, the eigenvalues of the Hessian of $e$ at singular
points, the $C^2$ norm of $v$ and the constants $Z_0$, $\be_0$ and $\ka$
of Hypothesis \hypflatnest. Fix a two--body interaction $v$ and a constant
$A>0$. Denote by $\cE_A$ the set of dispersion relations such that
$\cst{}{{\rm dKde}}(e,v)\le A$. 
If the dispersion relations $e$, $e'$ and all interpolants
 $(1-t)e+te'$, $0\le t\le 1$ are in $\cE_A$, 
and if $|\la|<\sfrac{1}{A}$, then
\begin{equation}\label{eq:injective}
e+\sum_{r=1}^R \la^r K_r(e)
=e'+\sum_{r=1}^R \la^r K_r(e')
\implies e=e'
\end{equation}
\end{mainproof}\goodbreak

\appendix

\section{Bounding General Diagrams - A Review}\label{ap:A}

For the convenience of the reader, we here provide a review of the 
general diagram bounding technique of \cite{FST1}. As a concrete example
of the technique,  we consider models in $d\ge 2$ for which the
interaction $v$ has $C^1$ Fourier transform and
the dispersion relation $e$ 
and its Fermi surface $\cF=\set{\bk}{e(\bk)=0}$ obey

\begin{itemize}

\item[\Hap] $\set{\bk}{|e(\bk)|\le 1}$ is compact. 

\item[\Hbp] $e(\bk)$ is $C^1$.

\item[\Hcp] $e(\tilde\bk)=0$ and $\bnabla e(\tilde\bk)=\bzer$ simultaneously
      only for finitely many $\tilde\bk$'s, called singular points.

\item[\Hdp] If $\tilde\bk$ is a singular point then         
            $\big[\sfrac{\partial^2\hfill}{\partial\bk_i\partial\bk_j} e(\tilde\bk)
                \big]_{1\le i,j\le d}$ is nonsingular.
\end{itemize}
and we prove that, any graph contributing to the proper self--energy
is $C^s$ for any $s<1$. Note that, in this appendix, we do {\it not} 
require the no--nesting condition of Hypothesis \hypflatnest.
The same methods apply to graphs with more than
two legs as well.

Let $G$ be any two--legged 1PI graph.
We also use the symbol $G$ to stand for the value of the graph $G$.
Singularities of the Fermi surface have no influence on the ultraviolet
regime, so we introduce a fixed ultraviolet cutoff by choosing a compactly
supported $C^\infty$ function $U(k)$ that is identically one on a neighbourhood
of $\{0\}\times\cF$ and use the propagator $C(k)=\sfrac{U(k)}{ik_0-e(\bk)}$.
If $M>1$ and $f$ is a suitable $C_0^\infty$ function that is supported
on $[M^{-4},1]$, we have the partition of unity \cite[\S 2.1]{FST1}
$$
U(k)=\sum_{j< 0}f(M^{-2j}|ik_0-e(\bk)|^2)
$$
and hence
\begin{equation}\label{eq:Cexpn}
C(k)=\sum_{j< 0}C_j(k)
\qquad\hbox{where}\qquad
C_j(k)=\sfrac{f(M^{-2j}|ik_0-e(\bk)|^2)}{ik_0-e(\bk)}
\end{equation}
Note that $f(M^{-2j}|ik_0-e(\bk)|^2)$ and $C_j(k)$ vanish unless 
$$
M^{j-2}\le |ik_0-e(\bk)|\le M^j
$$ 

First, suppose that $G$ is not renormalized.
Expand each propagator of $G$ using $\ C=\sum_{j< 0}C_j$ to give
$$
G=\sum_J G^J
$$
The sum runs over all possible labellings of the graph $G$, with each
labeling consisting of an assignment $J=\set{j_\ell< 0}{\ell\in G}$ 
of scales to the lines of $G$.
We now construct a natural hierarchy of subgraphs of $G^J$. This family
of subgraphs will be a forest, meaning that if $G_f,G_{f'}$ are in the
forest and intersect, either $G_f\subset G_{f'}$ or $G_{f'}\subset G_f$. 
First let, for each $j< 0$,
$$
G^{(\ge j)}=\set{\ell\in G^J}{j_\ell\ge j}
$$
be the subgraph of $G^J$ consisting of all lines of scale at least $j$.
(Think of an interaction line as a generalised four--legged vertex 
$$
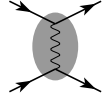
\figput{fourvert}
$$
rather than a line.)
There is no need for $G^{(\ge j)}$ to be connected. The forest $t(G^J)$
is the set of all connected subgraphs of $G^J$ that are components
of some $G^{(\ge j)}$. This forest is naturally partially ordered by
containment. In order to make $t(G^J)$ look like a tree with its root at the 
bottom, we define, for $f,f'\in t(G^J)$, $f>f'$ if $G_f\subset G_{f'}$.
We denote by $\pi(f)$ the highest fork of $t(G^J)$ below $f$
and by $\phi$ the root element, i.e. the element with $G_\phi=G$. To each
$G_f\in t(G)$ there is naturally associated the scale 
$j_f=\min\set{j_\ell}{\ell\in G_f}$. In the example below $j_4>j_3>j_1$
and $j_2>j_1$. External lines are in gray while internal lines are in black.
\vskip-0.6in
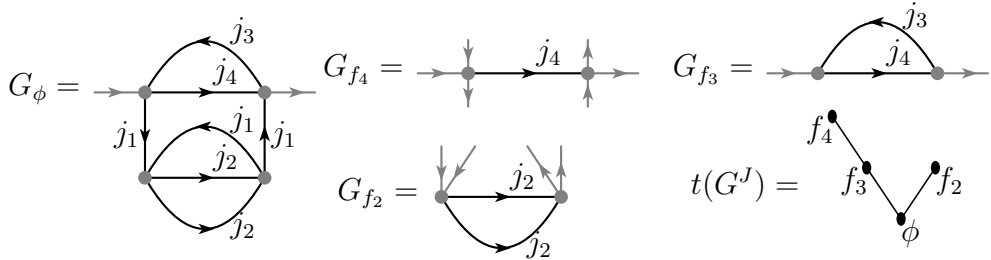
$$
\hskip-10pt
\figput{tgraph}\hskip-10pt\figplace{tgraphF4}{0 in}{0.1 in}
\figplace{tgraphF3}{0 in}{0.1 in}
\figplace{tgraphF2}{-3.8in}{-0.1 in}
\figplace{tree}{0.6in}{0.5 in}
$$
\vskip-0.2in
Reorganize the sum over $J$ using
\begin{equation}\label{eq:Gsum}
G=\sum_{j< 0}\sum_{t\in\cF(G)}\prod_{f\in t}\sfrac{1}{b_f!}
        \sum_{J\in\cJ(j,t,G)} G^J
\end{equation}
where
\begin{eqnarray*}
\cF(G)&=&\hbox{ the set of forests of subgraphs of }G\nonumber\\
b_f&=&\hbox{ the number of upward branches at the fork $f$ }\nonumber\\
\cJ(j,t,G)&=&\set{{\rm labellings\ }J{\rm\ of\ }G}{t(G^J)=t,\ j_\phi=j}
\nonumber\\
\end{eqnarray*}
A given labeling $J$ of $G$ is in $\cJ(j,t,G)$ if and only if
\begin{itemize}
\item 
for each $f\in t$, all lines of $G_f\setminus\cup_{f'\in t\atop
f'>f}G_{f'}$ have the same scale. Call the common scale $j_f$.
\item 
if $f>f'$ then $j_f>j_{f'}$
\item 
$j_\phi=j$
\end{itemize}
It is a standard result \cite[(2.72)]{FST1} that renormalization of
the dispersion relation may be implemented by modifying \Ref{eq:Gsum} 
as follows.
\begin{itemize}
\item
Each $\emptyset\ne f\in t$ for which $G_f$ has two external lines
is assigned a ``renormalization label''. This label can take the values $r$ 
and $c$. The set of possible assignments of renormalization labels, i.e.
 the set of all maps from $\set{f\in t}{G_f\hbox{ has
two external legs}}$ to $\{r,c\}$, is denoted $\cR(t)$.
\item
In the definition of the renormalized value of the graph $G$,
the value of each subgraph $G_f$ with renormalization label $r$ is replaced
by $(\bbbone-\ell)G_f(k)$. Here $\ell$ is the localization operator, which
we take\footnote{The main property that the localization operator should 
have is that $\sfrac{(\bbbone-\ell)G_f(k)}{ik_0-e(\bk)}$ should be bounded for
any (sufficiently smooth) $G_f$. Here is another possible localization operator
for $d=2$. In a neighbourhood of a regular point of the Fermi surface,
define $\ell G_f(k)=G_f(k_0=0,P\bk)$ where $P\bk$ is any reasonable
projection of $\bk$ onto the Fermi surface, as in \cite[Section 2.2]{FST1}.
In a neighbourhood of a singular point, use a coordinate system
in which $e(x,y)=xy$ and, in this coordinate system, define $\ell G_f(k_0,x,y)
=G_f(0,x,0)+G_f(0,0,y)-G_f(0,0,0)$. Use a partition of unity to
patch the different neighbourhoods together.}
 to be simply evaluation at $k_0=0$. For these $r$--forks, 
the constraint $j_f>j_{\pi(f)}$ still applies.  
\item
 In the definition of the renormalized value of the graph $G$,
the value of each subgraph $G_f$ with renormalization label $c$ is replaced
by $\ell G_f(k)$. For these $c$ forks the constraint $j_f>j_{\pi(f)}$
is replaced by $j_f\le j_{\pi(f)}$.
\end{itemize}
Given a graph $G$, a forest $t$ of subgraphs of $G$ and an assignment
$R$ of renormalization labels to the two--legged forks of $t$, we define $\cJ(j,t,R,G)$ to be the set of all assignments
of scales to the lines of $G$ obeying
\begin{itemize}
\item
for each $f\in t$, all lines of $G_f\setminus\cup_{f'\in t\atop
f'>f}G_{f'}$ have the same scale. Call the common scale $j_f$.
\item
if $G_f$ is not two--legged then $j_f>j_{\pi(f)}$
\item
if $G_f$ is two--legged and $R_f=r$, then $j_f>j_{\pi(f)}$
\item
if $G_f$ is two--legged and $R_f=c$ then $j_f\le j_{\pi(f)}$
\item
 $j_\phi=j$
\end{itemize}
Then, the value of the graph $G$ with all two--legged subdiagrams
correctly renormalized is
\begin{equation}\label{eq:Gren}
G=\sum_{j< 0}\ \sum_{t\in\cF(G)}\prod_{f\in t}\sfrac{1}{b_f!} \sum_{R\in\cR(t)}\ \sum_{J\in\cJ(j,t,R,G)} G^J
\end{equation}
To derive bounds on $G$, when we are not interested in the dependence
of those bounds on $G$ and in particular on the order of perturbation theory,
it suffices to derive bounds on $\sum_{j\le 0}\ \sum_{J\in\cJ(j,t,R,G)} G^J$
for each fixed $t$ and $R$.

\begin{proposition}\label{le:fixedrootscale}
Assume that the interaction has $C^1$ Fourier transform and
the dispersion relation obeys \Hap--\Hdp\ above. Let $G$ be any
two--legged 1PI graph of order $n$.
Let $t$ be a tree corresponding to a forest of subgraphs of $G$. 
Let $R$ be an assignment of $r,c$ labels to all forks $f>\phi$ of $t$ for 
which $G_f$ is two--legged. Let $\cJ(j,t,R,G)$ be the set of all assignments of 
scales to the lines of $G$ that have root scale $j$ and are consistent with 
$t$ and $R$. Let $s\in(0,1)$.
Then there is a constant  $\const$, depending on $s$ but independent of $j$, such that
\begin{eqnarray*}
\sum_{J\in \cJ(j,t,R,G)}\sup_{q}\big|G^J(q)\big|
&\le&\cst{}{n} |j|^{3n-2} M^{j}\\
\sum_{J\in \cJ(j,t,R,G)}\sup_{q}\big|\partial_{q_0}G^J(q)\big|
&\le&\cst{}{n} |j|^{3n-2} \\
\sum_{J\in \cJ(j,t,R,G)}\sup_{q,\bp}
                \sfrac{1}{|\bp|^s}\big|G^J(q+\bp)-G^J(q)\big|
&\le&\cst{}{n} |j|^{3n-2} M^{(1-s)j}
\end{eqnarray*}
\end{proposition}

\begin{remark}
Note that here the root scale is not summed over and $G_\phi$ is
not renormalized.  But all internal scales are summed over
and internal two--legged subgraphs that correspond
to $r$ and $c$ forks are renormalized and localized respectively.
\end{remark}

\begin{proof}
The proof is by induction on the depth of the graph, which is defined by
$$
D=\max\set{n}{\exists {\rm\ forks\ }f_1>\cdots>f_n>\phi{\rm\ with\ }
G_{f_1},\cdots,G_{f_n}\hbox{ all two--legged}}
$$
The inductive hypothesis is that
\begin{eqnarray*}
\sum_{J\in \cJ(j,t,R,G)}\sup_{q}\big|\partial^{s_0}_{q_0}G^J(q)\big|
&\le&\cst{}{n} |j|^{3n-2} M^{(1-s_0)j}\\
\sum_{J\in \cJ(j,t,R,G)}\sup_{q,\bp}
                \sfrac{1}{|\bp|^s}
\big|\partial^{s_0}_{q_0}G^J(q+\bp)-\partial^{s_0}_{q_0}G^J(q)\big|
&\le&\cst{}{n} |j|^{3n-2} M^{(1-s-s_0)j}
\end{eqnarray*}
for $s_0=0,1$ and all $s\in(0,1)$ (with the constant depending on $s$).

If $D>0$, decompose the tree $t$ into a pruned tree 
$\tilde t$ and insertion subtrees $\tau^1,\cdots,\tau^m$ by cutting
the branches beneath  all minimal $E_f=2$ forks $f_1,\cdots,f_m$. In other
words each of the forks $f_1,\cdots,f_m$ is an $E_f=2$ fork having
no $E_f=2$ forks, except $\phi$, below it in $t$. Each $\tau_i$
consists of the fork $f_i$ and all of $t$ that is above $f_i$. It has depth
at most $D-1$ so the corresponding subgraph $G_{f_i}$ obeys
the conclusion of this Proposition. Think of each subgraph $G_{f_i}$ as
a generalized vertex in the graph $\tilde G=G/\{G_{f_1},\cdots,G_{f_m}\}$.
Thus $\tilde G$ now has two-- as well as four--legged vertices. These
two--legged vertices have kernels of the form
\begin{equation}\label{cforkfn}
T_i(k)=\sum_{j_{f_i}\le j_{\pi(f_i)}}\ell G_{f_i}(k)
\end{equation}
when $f_i$ is a $c$--fork and of the form
\begin{equation}\label{rforkfn}
T_i(k)=\sum_{j_{f_i}> j_{\pi(f_i)}}(\bbbone-\ell)G_{f_i}(k)
\end{equation}
when $f_i$ is an $r$--fork. At least one of the external lines of $G_{f_i}$ 
must be of scale precisely $j_{\pi(f_i)}$ so
the momentum $k$ passing through $G_{f_i}$ lies in the support of 
$C_{j_{\pi(f_i)}}$. In the case of a $c$--fork $f=f_i$ we have, by
the inductive hypothesis,
\begin{eqnarray}
&&\hskip-30pt\sum_{j_{f}\le j_{\pi(f)}}\sum_{J_f\in\cJ(j_f,t_f,R_f,G_f)}
\sup_{k}\Big|\ell G_{f}^{J_f}(k)\Big|
\le\sum_{j_{f}\le j_{\pi(f)}}\sum_{J_f\in\cJ(j_f,t_f,R_f,G_f)}
\sup_{k}\left|G_{f}^{J_f}(k)\right|\nonumber\\
&&\hskip70pt
\le\sum_{j_{f}\le j_{\pi(f)}}\cst{}{n_f}|j_f|^{3n_f-2}M^{j_f}\nonumber\\
&&\hskip70pt
\le\cst{}{n_f}M^{j_{\pi(f)}}\sum_{i\ge 0}
                            (|j_{\pi(f)}|+i)^{3n_f-2}M^{-i}\nonumber\\
&&\hskip70pt
\le\cst{}{n_f}|j_{\pi(f)}|^{3n_f-2}M^{j_{\pi(f)}}\sum_{i\ge 0}
                            (i+1)^{3n_f-2}M^{-i}\nonumber\\
&&\hskip70pt
\le \cst{}{n_f}|j_{\pi(f)}|^{3n_f-2}M^{j_{\pi(f)}}\label{eq:cforkA} 
\end{eqnarray}
Here $n_f$ is the number of vertices of $G_f$ and $t_f$ and $R_f$ are 
the restrictions of $t$ and $R$ respectively to forks $f'\ge f$. Hence 
$J_f$ runs over all assignments of scales to the lines of $G_f$ consistent 
with the original $t$ and $R$ and with the specified value of $j_f$. 
Similarly,
\begin{eqnarray}
&&\hskip-0.7in\sum_{j_{f}\le j_{\pi(f)}}\sum_{J_f\in\cJ(j_f,t_f,R_f,G_f)}
\sup_{k,\bp}\sfrac{1}{|\bp|^s}\Big|\ell G_{f}^{J_f}(k+\bp)-\ell G_{f}^{J_f}(k)\Big|
\nonumber\\
&&\hskip2in\le\cst{}{n_f}|j_{\pi(f)}|^{3n_f-2}M^{(1-s)j_{\pi(f)}}
\label{eq:cforkB} 
\end{eqnarray}
Note that $\ell G_{f}^{J_f}(k)$ is independent of $k_0$ so that 
$\partial_{k_0}$ may never act on it.

In the case of an $r$--fork $f=f_i$, we have
$$
|(\bbbone-\ell)G(k)|=\big|G(k_0,\bk)-G(0,\bk)\big|
\le |k_0|\sup_k\big|\partial_{k_0} G(k)\big|
$$ 
Hence,
by the inductive hypothesis and, when $s_0=0$, the mean value theorem,
\begin{eqnarray}
&&\hskip-20pt\sum_{j_{f}> j_{\pi(f)}}\sum_{J_f\in\cJ(j_f,t_f,R_f,G_f)}
\sup_{k}1\big(C_{j_{\pi(f)}}(k)\ne 0\big)
\Big|\partial_{k_0}^{s_0}(\bbbone-\ell)G_{f}^{J_f}(k)\Big|\nonumber\\
&&\hskip70pt
\le \sum_{j_{f}> j_{\pi(f)}}\ \sum_{J_f\in\cJ(j_f,t_f,R_f,G_f)}
M^{(1-s_0)j_{\pi(f)}}\sup_{k}\Big|\partial_{k_0} G_{f}^{J_f}(k)\Big|\nonumber\\
&&\hskip70pt
\le\cst{}{n_f}M^{(1-s_0)j_{\pi(f)}}
\sum_{j_{f}> j_{\pi(f)}}|j_{f}|^{3n_f-2}\nonumber\\
&&\hskip70pt
\le\cst{}{n_f}|j_{\pi(f)}|^{3n_f-1}M^{(1-s_0)j_{\pi(f)}}\label{eq:rforkA}
\end{eqnarray}
Similarly, for $|k_0|\le M^{j_{\pi(f)}}$,
\begin{eqnarray}
&&\hskip-20pt\sum_{j_{f}> j_{\pi(f)}}\ \sum_{J_f\in\cJ(j_f,t_f,R_f,G_f)}
\sup_{k,\bp}
\sfrac{1}{|\bp|^s}
\Big|\partial^{s_0}_{k_0}(\bbbone-\ell) G_{f}^{J_f}(k+\bp)-
\partial^{s_0}_{k_0}(\bbbone-\ell) G_{f}^{J_f}(k)\Big|\nonumber\\
&&\hskip-10pt
\le\sum_{j_{f}> j_{\pi(f)}}\ \sum_{J_f\in\cJ(j_f,t_f,R_f,G_f)}\hskip-22pt
M^{(1-s_0)j_{\pi(f)}}\sup_{k,\bp}
\sfrac{1}{|\bp|^s}
\Big|\partial_{k_0}G_{f}^{J_f}(k+\bp)-
\partial_{k_0}G_{f}^{J_f}(k)\Big|\nonumber\\
&&\hskip-10pt
\le\cst{}{n_f}M^{(1-s_0)j_{\pi(f)}}
\sum_{j_{f}> j_{\pi(f)}}|j_{f}|^{3n_f-2}M^{-j_fs}\nonumber\\
&&\hskip-10pt
\le\cst{}{n_f}|j_{\pi(f)}|^{3n_f-2}M^{(1-s_0-s)j_{\pi(f)}}\label{eq:rforkB}
\end{eqnarray}

We are now ready to bound $\tilde G^{\tilde J}$, where $\tilde J$ is the
restriction of $J$ to $\tilde G$.
It is both convenient and standard to get rid of the conservation of 
momentum delta functions arising in the value of $\tilde G^{\tilde J}$ 
by integrating out some momenta. Then, instead of having one  
$(d+1)$--dimensional integration variable $k$ for each line of the diagram, there is one for each momentum loop. Here is a convenient way to select 
these loops. Pick any spanning tree $\tilde T$ for $\tilde G$. A spanning tree 
is a subgraph of $\tilde G$ that is a tree and contains all the vertices of 
$\tilde G$. We associate to each line $\ell$ of $\tilde G\setminus \tilde T$ 
the ``internal momentum loop'' $\La_\ell$ that consists of $\ell$ and the 
unique path in $\tilde T$ joining the ends of $\ell$. The ``external 
momentum path'' is the unique path in $\tilde T$ joining the external legs. 
It carries the external momentum $q$. The loop $\La_\ell$ carries momentum 
$k_\ell$. The momentum $k_{\ell'}$ of each line $\ell'\in \tilde T$ is the 
signed sum of all loop and external momenta passing through $\ell'$.

The form of the integral giving the value of $\tilde G^{\tilde J}(q)$ is then
\begin{equation}\label{eq:tildeGtildeJform}
\tilde G^{\tilde J}(q)=\int\prod_{\ell\in \tilde G\setminus \tilde T}
\dbar k_\ell \prod_{\ell\in \tilde G}C_{j_\ell}(k_\ell)\prod_\rv u_\rv
\qquad\hbox{ where } \dbar k = \sfrac{d^{d+1}k}{(2\pi)^{d+1}}
\end{equation}
Here $\tilde T$ is any spanning tree for $\tilde G$. The loops are labeled 
by the lines of $\tilde G\setminus \tilde T$. For each $\ell\in \tilde T$,
the momentum $k_\ell$ is a signed sum of loop momenta and external momentum
$q$. The product $\prod_\rv$ runs over the vertices of $\tilde G$
and $u_\rv$ is the vertex function for $\rv$. If $\rv$ is one of the original
interaction vertices then $u_\rv$ is just $v$ evaluated at the 
signed sum of loop and external momenta passing through $v$. If $\rv$ is
a two--legged vertex, then $u_\rv$ is given either by \Ref{cforkfn} or by
\Ref{rforkfn}. 

We are now ready to bound $\tilde G$ in six steps.
\begin{enumerate} 
\item
Choose a spanning tree $\tilde T$ for $\tilde G$ with the property that 
$\tilde T\cap \tilde G^{\tilde J}_f$ is a connected tree for every 
$f\in t(\tilde G^{\tilde J})$. $\tilde T$ can be built up inductively, 
starting with the smallest subgraphs $\tilde G_f$, because, by 
construction, every $\tilde G_f$ is connected and $t(\tilde G^J)$ is a forest.
Such a spanning tree is illustrated below for the example given just before
\Ref{eq:Gsum} with $j_4>j_3>j_1$, $j_2>j_1$.
$$
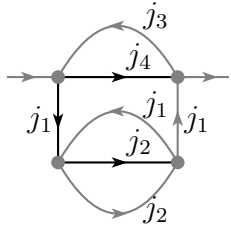
\figput{span}
$$
\item
 Apply any $q$--derivatives. By the product rule, or,
in the case of a ``discrete derivative'', the ``discrete product rule''
\begin{eqnarray*}
&&\hskip-20ptf(k+\bq)g(k'+\bq)-f(k)g(k')\\
&&\hskip20pt=\big[f(k+\bq)-f(k)\big]g(k')
+f(k+\bq)\big[g(k'+\bq)-g(k')\big],
\end{eqnarray*}
each derivative may act on any line or vertex on the ``external momentum
path''. It suffices to consider any one such action.
\item
 Bound each two--legged renormalized subgraph (i.e.
$r$--fork) by (\ref{eq:rforkA},\ref{eq:rforkB}) and each two--legged 
counterterm (i.e. $c$--fork) by (\ref{eq:cforkA},\ref{eq:cforkB}). 
Observe that when $s_0$ $k_0$--derivatives and $s$ $\bk$--derivatives act 
on the vertex, the bound is no worse than $M^{-(s_0+s)j}$ times the bound with no derivatives,
because we necessarily have $j\le j_{\pi(f)}<0$. 
\item
 Bound all remaining vertex functions, $u_\rv$, (suitably differentiated)
by their suprema in momentum space. 
\item
 Bound each  propagator 
\begin{eqnarray}\label{eq:propbnd}
&&\hskip-25pt\big|\partial^{s_0}_{k_0}C_{j_\ell}(k)\big|
\le\const M^{-(1+s_0)j_\ell}\ 1\big(|ik_0-e(\bk)|\le M^{j_\ell}\big)
\nonumber\\
&&\hskip-25pt\sfrac{1}{|\bp|^s}
\big|\partial^{s_0}_{k_0}C_{j_\ell}(k+\bp)
            -\partial^{s_0}_{k_0}C_{j_\ell}(k)\big|
\le\const M^{-(1+s_0+s)j_\ell}\nonumber
\nonumber\\
\end{eqnarray}
Once again,  when $s_0$ $k_0$--derivatives and $s$ $\bk$--derivatives act 
on the propagator, the bound is no worse than $M^{-(s_0+s)j}$ times the bound 
with no derivatives, because we necessarily have $j\le j_\ell<0$. 
 We now have $\big|\partial_{k_0}^{s_0}\tilde G^{\tilde J}(q)\big|$
and $\sfrac{1}{|\bp|^s}\Big|\partial_{k_0}^{s_0}\tilde G^{\tilde J}(q+\bp)
-\partial_{k_0}^{s_0}\tilde G^{\tilde J}(q)\Big|$ bounded, uniformly in
$q$ and $\bp$ by
\begin{eqnarray*}
&&\hskip-20pt\cst{}{n}M^{-(s_0+s)j}
   \prod_{i=1}^m|j_{\pi(f_i)}|^{3n_{f_i}-1}M^{j_{\pi(f_i)}}
    \prod_{\ell\in \tilde G}M^{-j_\ell} \\
&&\hskip1in\int\prod_{\ell\in \tilde G\setminus \tilde T}\dbar k_\ell
\ \prod_{\ell\in \tilde G\setminus \tilde T}1\big(
|ik_{\ell0}-e(\bk_\ell)|\le  M^{j_\ell}\big)
\end{eqnarray*}
with $s=0$ in the first case. We remark that for the bound on 
$\big|\partial_{k_0}^{s_0}\tilde G^{\tilde J}(q)\big|$ we may replace 
the $\prod_{\ell\in \tilde G\setminus \tilde T}$ in
$\prod_{\ell\in \tilde G\setminus \tilde T}1\big(
                 |ik_{\ell0}-e(\bk_\ell)|\le  M^{j_\ell}\big)$ by
$\prod_{\ell\in \tilde G}$. These extra integration constraints are not
used in the current naive bound, but are used in other bounds that exploit
``overlapping loops''.
\item
Integrate over the remaining loop momenta. 
Integration over $k_\ell$ with $\ell\in \tilde G\setminus \tilde T$ is 
bounded by vol$\set{k_{0\ell}}{|k_{0\ell}|\le M^{j_\ell}}\le 2M^{j_\ell}$ times the volume of
$\set{\bk_\ell}{|e(\bk_\ell)|\le M^{j_\ell}}$, which, by
Lemma \ref{le:2.4}, is bounded by a constant times $|j_\ell|M^{j_\ell}$.
\end{enumerate} 

\noindent The above six steps give that  
$\big|\partial_{k_0}^{s_0}\tilde G^{\tilde J}(q)\big|$
and $\sfrac{1}{|\bp|^s}\Big|\partial_{k_0}^{s_0}\tilde G^{\tilde J}(q+\bp)
-\partial_{k_0}^{s_0}\tilde G^{\tilde J}(q)\Big|$ are bounded, uniformly in
$q$ and $\bp$ by
\begin{equation}\label{eq:step6bnd}
B^{\tilde J}=\cst{}{n}M^{-(s_0+s)j}
   \prod_{i=1}^m|j_{\pi(f_i)}|^{3n_{f_i}-1}M^{j_{\pi(f_i)}}
    \prod_{\ell\in \tilde G}M^{-j_\ell}
\prod_{\ell\in \tilde G\setminus  \tilde T}|j_\ell|  M^{2j_\ell}
\end{equation}
again with $s=0$ in the first case.
Define the notation
\begin{eqnarray*}
\tilde T_f&=&\hbox{number of lines of } \tilde T\cap \tilde G_f\\
\tilde L_f&=&\hbox{number of internal lines of } \tilde G_f\\
n_f&=&\hbox{number of vertices of } G_f\\
E_f&=&\hbox{number of external lines of } G_f\\
E_\rv&=&\hbox{number of lines hooked to vertex } \rv\\
\end{eqnarray*}
Applying
$$
M^{\al j_\ell}
=M^{\al j_\phi}\prod_{f\in t\atop{f>\phi\atop \ell\in G_f}}
M^{\al(j_f-j_{\pi(f)})}
$$
to each $M^{-j_\ell}$  and $M^{2j_\ell}$ and 
\begin{eqnarray*}
M^{j_{\pi(f_i)}}
&=&M^{j_\phi}\prod_{f\in \tilde t\atop \phi<f<f_i}
M^{j_f-j_{\pi(f)}}\nonumber\\
&=&M^{-{1\over 2}(E_{f_i}-4)j_\phi}
\prod_{f\in \tilde t\atop{f>\phi\atop f_i\in \tilde G_f}}
M^{-{1\over 2}(E_{f_i}-4)(j_f-j_{\pi(f)})}
\end{eqnarray*}
for each $1\le i\le m$ (thinking of $f_i$ as a vertex of $\tilde G$) gives
\begin{eqnarray*}
B^{\tilde J}&\le&\cst{}{n} M^{-(s_0+s)j}
|j|^{\tilde L_\phi-\tilde T_\phi+\sum (3n_{f_i}-1)}M^{j(\tilde L_\phi-2\tilde T_\phi
-\sum_{\rv\in \tilde G}{1\over 2}(E_\rv-4))}\\
&&\hskip2in
\prod_{f\in \tilde t\atop f>\phi}M^{(j_f-j_{\pi(f)})(\tilde L_f-2\tilde T_f
-\sum_{\rv\in\tilde G_f}{1\over 2}(E_\rv-4))}
\end{eqnarray*}
The sums $\sum_{v\in\tilde G}$  and $\sum_{v\in\tilde G_f}$ run over two--
as well as four--legged generalized vertices. As
\begin{eqnarray*}
\tilde L_f
&=&\sfrac{1}{2}\big(\sum_{\rv\in \tilde G_f}E_\rv-E_f\big)\qquad\hbox{and}
\qquad
\tilde T_f=\sum_{\rv\in \tilde G_f}1 -1\\
\Longrightarrow 
\tilde L_f-2\tilde T_f
       &=& \sfrac{1}{2}\big(4-E_f+\sum_{\rv\in \tilde G_f}(E_\rv-4)\big) \\
\end{eqnarray*}
and we have
\begin{eqnarray}\label{eq:pwrcntbnd}
B^{\tilde J}
&\le&\cst{}{n} M^{-(s_0+s)j}
|j|^{\tilde L_\phi-\tilde T_\phi+\sum (3n_{f_i}-1)}M^{{1\over 2}j(4-E_\phi)}
\prod_{f\in \tilde t\atop f>\phi}M^{{1\over 2}(j_f-j_{\pi(f)})(4-E_f)}
\nonumber\\
&=&\cst{}{n} M^{-(s_0+s)j}
|j|^{\tilde L_\phi-\tilde T_\phi+\sum (3n_{f_i}-1)}M^j
\prod_{f\in \tilde t\atop f>\phi}M^{{1\over 2}(j_f-j_{\pi(f)})(4-E_f)}
\end{eqnarray}
since $E_\phi=2$. The scale sums are performed by repeatedly applying
\begin{equation}\label{eq:scalesum}
\sum_{j_f\atop j_f>j_{\pi(f)}}M^{{1\over 2}(j_f-j_{\pi(f)})(4-E_f)}
\le\cases{|j|&if $E_f=4$\cr
           \sfrac{1}{M-1}& if $E_f>4$\cr}
\end{equation}
starting with the highest forks, and give at most $\tilde
L_\phi-1$ additional factors of $|j|$ since
$$
\#\{f\in t(\tilde G^{\tilde J}),\ f\ne \phi\}\le \tilde L_\phi-1
$$ 
Thus
\begin{eqnarray}
\sum_{\tilde J\in \cJ(j,\tilde t,\tilde G)}B^{\tilde J}
&\le& \cst{}{n} 
|j|^{2\tilde L_\phi-\tilde T_\phi-1+\sum (3n_{f_i}-1)}M^{(1-s_0-s)j}\nonumber\\
&\le& \cst{}{n} |j|^{3n-2}M^{(1-s_0-s)j}\label{eq:scalesumbnd}
\end{eqnarray}
since, using $\tilde n_4$ to denote the number of four--legged vertices
in $\tilde G$,
\begin{eqnarray*}
&&\hskip-30pt2\tilde L_\phi-\tilde T_\phi-1+\sum_{i=1}^m (3n_{f_i}-1)\cr
&&\hskip20pt=2\sfrac{1}{2}\big(4\tilde n_4+2m-2\big)-\big(\tilde n_4+m-1\big)
-1+3\sum_{i=1}^m n_{f_i}-m\\
&&\hskip20pt=3\tilde n_4+3\sum_{i=1}^m n_{f_i}-2\\
&&\hskip20pt=3n-2
\end{eqnarray*}
This is the desired bound.
\end{proof}

\begin{corollary}\label{le:psecont}
Assume that the interaction has $C^1$ Fourier transform and
the dispersion relation obeys \Hap--\Hdp\ above.  
Let $G(q)$ be any graph contributing to the
proper self--energy.
Then, for every $0<s<1$, 
\begin{eqnarray*}
\sup_{q}|G(q)|&<&\infty\cr
\sup_{q,\bp}\sfrac{1}{|\bp|^s}|G(q+\bp)-G(q)| &<&\infty\cr
\end{eqnarray*}
\end{corollary}
\begin{proof}
Both bounds are immediate from Proposition \ref{le:fixedrootscale}. One merely 
has to sum over $j,t$ and $R$. The bound on $\sup_{q}|G(q)|$ was also proven 
by these same methods in \cite{Brox}.
\end{proof}

\end{document}

%% file: lemsmac.tex
\newcommand{\bbbone}{{\mathchoice {\rm 1\mskip-4mu l} {\rm 1\mskip-4mu l}    {\rm 1\mskip-4.5mu l} {\rm 1\mskip-5mu l}}}

\newcommand{\bC}{{\mathbb C}}

\newcommand{\bR}{{\mathbb R}}

\newcommand{\al}{\alpha}
\newcommand{\be}{\beta}
\newcommand{\ga}{\gamma}
\newcommand{\de}{\delta}
\newcommand{\ep}{\epsilon}
\newcommand{\et}{\eta}
\newcommand{\ze}{\zeta}
\renewcommand{\th}{\theta} 

\newcommand{\ka}{\kappa}
\newcommand{\la}{\lambda}
\newcommand{\rh}{\rho}

\newcommand{\om}{\omega}

\newcommand{\veps}{\varepsilon}
\newcommand{\De}{\Delta}

\newcommand{\La}{\Lambda}
\newcommand{\Si}{\Sigma}

\newcommand{\bGa}{{{\rm I}\kern-.16em \Gamma}}
\newcommand{\cA}{{\cal A}}

\newcommand{\cD}{{\cal D}}
\newcommand{\cE}{{\cal E}}
\newcommand{\cF}{{\cal F}}

\newcommand{\cJ}{{\cal J}}
\newcommand{\cK}{{\cal K}}

\newcommand{\cM}{{\cal M}}

\newcommand{\cO}{{\cal O}}

\newcommand{\cR}{{\cal R}}

\newcommand{\cT}{{\cal T}}

\newcommand{\cV}{{\cal V}}

\newcommand{\cX}{{\cal X}}

\newcommand{\del}{\partial}

\newcommand{\abs}[1]{{\left\vert #1 \right\vert}}

\newcommand{\Ref}[1]{$(\ref{#1})$}
\newcommand{\sfrac}[2]{{\textstyle \frac{#1}{#2}}}

\newcommand{\vol}{{\hbox{ vol }}}
\newcommand{\dbar}{{\mkern2mu\mathchar'26\mkern-2mu\mkern-9mud}}

\newcommand{\const}{\hbox{ \rm const }}

\newif\ifintrmk
\newcommand{\tst}{\textstyle}

%% file: figMac.tex
\def\suffix{ps}
\newcount\system
\global\system=3   

\def\ifundefined#1{\expandafter\ifx\csname#1\endcsname\relax}
\ifundefined{figdir}\def\figdir{}\fi
\newcount\firstline
\newdimen\pswidth  \newdimen\xleft
\newdimen\psheight \newdimen\ytop \newdimen\ybot
\newcount\justx \newcount\justy
\global\justx=0 \global\justy=0
\newdimen\vpos \newtoks\labeL 
\newread\labeLfile \newdimen\xcoord \newdimen\ycoord
\newif\ifdoit 
\newbox\labox
\newdimen\xdvikwid 
\newdimen\xdvikht
\newdimen\pspoints
\newdimen\rwi
\newcount\temp
\def\readdim#1{\global\read\labeLfile to \temp
\global #1=\temp pt}
\def\figcrop#1{\par
\openin\labeLfile=\figdir#1.lbl                                              
\global\read\labeLfile to\firstline\message{#1}               
\global\read\labeLfile to\temp
\readdim{\ybot}
\readdim{\xleft}
\readdim{\ytop}
\global\read\labeLfile to\justx
\global\read\labeLfile to\justy
\global\read\labeLfile to\labeL
\readdim{\pswidth}
\global\advance\pswidth by -\xleft
\readdim{\psheight}
\global\advance\ybot by -\psheight
\global\advance\psheight by -\ytop
\global\read\labeLfile to\justx
\global\read\labeLfile to\justy
\global\read\labeLfile to\labeL
\vbox to\psheight{\vfill
\ifnum\system=1
\fi 
\ifnum\system=2
\fi 
\ifnum\system=3
                                                 \fi         
\ifnum\system=4
\fi         
\ifnum\system=1
\hbox to \pswidth{\kern-\xleft\special{postscriptfile \figdir#1.\suffix }\hfil}\fi
\ifnum\system=2
\hbox to \pswidth{\kern-\xleft\special{ps: plotfile \figdir#1.\suffix }\hfil}\fi
\ifnum\system=3
\hbox to \pswidth{\kern-\xleft\includegraphics{\figdir#1.\suffix}\hfil}\fi
\ifnum\system=4
\hbox to \pswidth{\kern-\xleft\includegraphics{\figdir#1.\suffix}\hfil}\fi
\ifnum\system=5
\hbox to \pswidth{\kern-\xleft\includegraphics{\figdir#1.\suffix}\hfil}\fi 
\ifnum\system=6
   \xdvikwid=\pswidth
   \xdvikht=\psheight
   {\global\divide\xdvikwid by \pspoints}
   {\global\divide\xdvikht by \pspoints}
   \rwi=\xdvikwid
    {\global\multiply\rwi by 10}
\hbox to \pswidth{\kern-\xleft\includegraphics{\figdir#1.\suffix\space}\hfil}\fi                   
\vskip -\baselineskip
\vskip -\ybot 
\vskip-\psheight %
\hbox to\pswidth  {\hss}%
\parindent=0pt\offinterlineskip                                       
\vpos=0 pt%
\loop\readdim{\xcoord}                                 
\ifdim \xcoord < -999pt \doitfalse\else\doittrue\fi                        
\ifdoit \advance \xcoord by -\xleft
\readdim{\ycoord}
\advance \ycoord by -\ytop                              
\global\read\labeLfile to\justx                                       
\global\read\labeLfile to\justy                                       
\global\read\labeLfile to\labeL
\global\setbox\labox=\hbox{\labeL\hskip-0.3em}%
\advance\vpos by-\ycoord                                              
\vskip-\vpos \vpos=\ycoord                                         
\hbox to\pswidth{\hskip\xcoord %
\hbox to 0pt{\ifnum\justx>0\hss\fi%
\vbox to0pt{%
\ifnum\justy<2\vss\fi%
\copy\labox\kern0pt%
\ifnum\justy>0\vss\fi}%
\ifnum\justx<2\hss\fi}%
\hss}%
\repeat%
\advance\vpos by-\psheight%
\vskip-\vpos %
}\closein\labeLfile}
\def\figplace#1#2#3{
\openin\labeLfile=\figdir#1.lbl
\ifeof \labeLfile
       \immediate\write16{***Can't find \figdir#1.lbl; Skipping it.***}
\else  \closein\labeLfile
       \null\hskip#2\raise #3 \hbox{\figcrop{#1}}
\fi
}
\def\figput#1{
\openin\labeLfile=\figdir#1.lbl
\ifeof \labeLfile
       \immediate\write16{***Can't find \figdir#1.lbl; Skipping it.***}
\else  \closein\labeLfile
       \hbox{\figcrop{#1}}
\fi
}

%% file: FS1.bbl
\begin{thebibliography}{99}

\bibitem{VanHove}
L.\ Van Hove, 
{\sl Phys. Rev. \bf 89} (1953) 1189.

\bibitem{Morse}
M.\ Morse, 
{\sl Trans.\ Am.\ Math.\ Soc.\ \bf 27} (1925) 345.


\bibitem{Morsebis}
J.\ Mawhin and M.\ Willem, 
{\sl Critical Point Theory and Hamiltonian Systems}, Springer--Verlag (1989).


\bibitem{Koma}
T.\ Koma and H.\ Tasaki, 
{\sl Phys. Rev. Lett. \bf 68} (1992) 3248. 

\bibitem{Markiewicz}
R.S.\ Markiewicz, 
{\sl J. Phys. Chem. Solids \bf 58} (1997) 1179. 

\bibitem{HM}
C.J.\ Halboth and W.\ Metzner,
 Phys. Rev. B {\bf 61} (2000) 7364,  
%
Phys. Rev. Lett. {\bf 85} (2000) 5162.

\bibitem{HSFR}
C.\ Honerkamp, M.\ Salmhofer, N.\ Furukawa,  T.M.\ Rice,
{\sl Phys.\ Rev.\ \bf B 63} (2001) 035109.

\bibitem{HS}
C.\ Honerkamp, M.\ Salmhofer,
{\sl Phys.\ Rev.\ Lett.\ \bf 87} (2001) 187004.

\bibitem{MeNeu}
A.\ Neumayr, W.\ Metzner, 
{\sl Phys. Rev. \bf B 67} (2003) 035112. 

\bibitem{MeQC}
W.\ Metzner, D.\ Rohe, S. Andergassen, 
{\sl Phys. Rev. Lett. \bf 91} (2003) 066402.

\bibitem{SFS2}
J.\ Feldman, M.\ Salmhofer, 
{\sl Singular Fermi Surfaces II. The Two--Dimensional Case},
in preparation.

\bibitem{FST1} J.\ Feldman, M.\ Salmhofer, E.\ Trubowitz, 
{\sl J.\ Stat.\ Phys.\ \bf 84} (1996) 1209.

\bibitem{FST2} J.\ Feldman, M.\ Salmhofer, E.\ Trubowitz, 
{\sl Comm.\ Pure Appl.\ Math.\ \bf 51} (1998) 1133.

\bibitem{FST3} J.\ Feldman, M.\ Salmhofer, E.\ Trubowitz, 
{\sl Comm.\ Pure Appl.\ Math.\ \bf 52} (1999) 273.

\bibitem{FST4} J.\ Feldman, M.\ Salmhofer, E.\ Trubowitz, 
{\sl Comm.\ Pure Appl.\ Math.\ \bf 53} (2000) 1350.

\bibitem{FKTasy}
J.\ Feldman, H.\ Kn\" orrer, E.\ Trubowitz,
{\sl Comm.\  PDE \bf 25} (2000), 319. 

\bibitem{FMRT}
J.\ Feldman, J.\ Magnen, V.\ Rivasseau, E.Trubowitz,
{\sl Helv.\ Phys.\ Acta \bf 65} (1992) 679.


\bibitem{Brox}
D.\ Brox, \textit{ Thesis (M.Sc.)--University of British Columbia}, 2005.


\end{thebibliography}
